\documentclass[cleveref,english,autoref,thm-restate,numberwithinsect,nameinlink]{lipics-v2021}

\hideLIPIcs

\nolinenumbers

\usepackage{hyperref}
\usepackage{amsmath}
\usepackage{amsthm}
\usepackage{amssymb}
\usepackage{enumerate}
\usepackage{thmtools}
\usepackage{mathtools}
\usepackage{graphicx}
\usepackage{multirow}
\usepackage{xcolor}
\usepackage{tikz}
\usetikzlibrary{decorations.pathmorphing, decorations.pathreplacing, decorations.shapes,shapes,positioning, fit,calligraphy}

\usepackage{framed}

\usepackage[]{todonotes}

\usepackage{nicefrac}
\usepackage{tabularx}
\usepackage{algorithm2e}
\usepackage{etoolbox}

\newtheorem{rrule}{Rule}

\definecolor{myred}{rgb}{1,0.6,0.6}
\definecolor{myblue}{rgb}{0,0,1}

\definecolor{mycolor1}{RGB}{100, 143, 255}
\definecolor{mycolor2}{RGB}{26, 255, 26}
\definecolor{mycolor3}{RGB}{220, 38, 127}
\definecolor{mycolor4}{RGB}{254, 97, 0}
\definecolor{mycolor5}{RGB}{255, 176, 0}

\newcommand{\kommentar}[1]{}
\newcommand{\Oh}{\ensuremath{\mathcal{O}}}

\newcommand{\probname}[1]{{\normalfont\textsc{#1}}}
\newcommand{\prob}[3]{
	\begin{nolinenumbers}
		\begin{framed}
			\noindent \probname{#1} \\
			\textbf{Input:} #2 \\
			\textbf{Question:} #3
		\end{framed}
	\end{nolinenumbers}
}

\newcommand{\PROB}[1]{{{\normalfont\textsc{#1}}}\xspace}

\newcommand{\CL}{\PROB{Clique}}
\newcommand{\CC}{\PROB{Correlation Clustering}}
\newcommand{\MC}{\PROB{Edge Multicut}}

\newcommand{\TSAT}{\PROB{3-SAT}}

\newcommand{\NP}{{\ensuremath{\mathrm{NP}}}\xspace}

\newcommand{\coNP}{{\ensuremath{\mathrm{coNP}}}\xspace}
\newcommand{\coNPpoly}{{\ensuremath{\mathrm{\coNP/poly}}}\xspace}

\newcommand{\bth}{$\NP\subseteq\coNPpoly$}

\newcommand{\N}{\ensuremath{\mathbb{N}}}
\newcommand{\Z}{\mathbb{Z}_{\ge 0}}

\newcommand{\mc}{\mathcal{C}}

\newcommand{\true}{\ensuremath{\mathtt{true}}}
\newcommand{\false}{\ensuremath{\mathtt{false}}}

\DeclareMathOperator{\cost}{cost}

\EventEditors{Jan Goedgebeur and Pawe{\l} Rz\k{a}\.{z}ewski}
\EventNoEds{2}
\EventLongTitle{52nd International Workshop on Graph-Theoretic Concepts in Computer Science (WG 2026)}
\EventShortTitle{WG 2026}
\EventAcronym{WG}
\EventYear{2026}
\EventDate{June 2--4, 2026}
\EventLocation{Kortrijk, Belgium}
\EventLogo{}
\SeriesVolume{376}
\ArticleNo{32}

\Copyright{Jaroslav Garvardt and Christian Komusiewicz}


\author{Jaroslav Garvardt}{Institute of Computer Science, University of Jena, Germany}{jaroslav.garvardt@uni-jena.de}{https://orcid.org/0000-0002-8762-8567}{Supported by the Carl Zeiss Foundation, Germany, within the project ‘‘Interactive Inference’’.}

\author{Christian Komusiewicz}{Institute of Computer Science, University of Jena, Germany}{c.komusiewicz@uni-jena.de}{https://orcid.org/0000-0003-0829-7032}{}

\authorrunning{J.~Garvardt and C.~Komusiewicz} 

\title{Clustering with Locally Bounded Ignorance}

\begin{document}

\maketitle
\begin{abstract}
  In \textsc{Correlation Clustering}, the input is a
  graph~$G=(V,E)$ with weight function~$\omega: {V \choose 2}\to \Z$
  and the task is to partition the vertex set into clusters such that
  the total weight of edges between clusters and missing edges
   inside clusters is minimized. Due to close connections
  between~\textsc{Correlation Clustering} and \textsc{Edge Multicut},
  deciding whether there is a partition with total cost at most~$k$ is
  FPT with respect to~$k$ but a polynomial kernel is presumably
  impossible. We study the influence of the structure of the fuzzy
  edge graph, that is, the graph induced by the weight-0 edges, on the
  problem complexity. We show in particular that \textsc{Correlation
    Clustering} admits a polynomial problem kernel when parameterized
  by~$k+d$, where~$d$ is the degeneracy of the fuzzy edge graph, and when
  parameterized by~$k+c$, where~$c$ is the closure of the fuzzy edge
  graph. We complement these positive results by showing hardness for
  several settings where the graph induced by the edges and nonedges has very restricted structure.
\end{abstract}

\newpage

\section{Introduction}
Clustering is the task of grouping similar objects together. A crucial tool for formalizing this task are similarity graphs where the objects are vertices and we draw an edge between two vertices if and only if these objects are similar. Clustering now corresponds to computing a partition of the vertex sets into clusters such that there are few edges between different clusters and many edges inside the clusters. There are different concrete formulations of this problem, with \textsc{Cluster Editing} being one of the most prominent ones~\cite{BB13}. Here, we aim to minimize the total number of misclassified vertex pairs which are defined as edges between different clusters and missing edges inside clusters. This can also be viewed as minimizing the number~$k$ of edge insertions and deletions to transform the input graph into a cluster graph, that is, a graph whose connected components are all cliques. These cliques correspond to the output clusters and the number~$k$ of necessary modifications is exactly the number of misclassifications. \textsc{Cluster Editing} is NP-hard~\cite{KM86,SST04} and cannot be solved in subexponential time under the ETH~\cite{KU12}. This hardness motivates the study from the viewpoint of parameterized algorithms with most studies focusing on FPT-algorithms or kernelizations for the parameter~$k$ and related parameters~\cite{Boeck12,GGHN05,Guo09,KU11,KU12,BFK18}.

In many applications, the similarity data is more precise than merely a binary predicate that distinguishes between similar and dissimilar vertex pairs. In particular, the edges and nonedges may be assigned a strictly positive integer weight, with large weights corresponding to a strong evidence for similarity and dissimilarity, respectively. This problem, known as \textsc{Weighted Cluster Editing}, also has FPT-algorithms and kernelizations for the parameter~$k$~\cite{BBBT09,CC12} but it is harder than \textsc{Cluster Editing} when the input is restricted to special graph classes~\cite{IKP25,GKM26}.

While being more realistic, \textsc{Weighted Cluster Editing} still assumes full knowledge in the sense that each vertex pair is either deemed to be similar or dissimilar. To address that for some vertex pairs there may be no information, one may assign weight 0 to them, which means that they can be ignored when computing the number of misclassifications or alternatively, that they can be modified for free. This most general version of the problem is known as \textsc{Correlation Clustering}~\cite{BBC04,DemaineEFI06}.

The additional freedom of ignoring some vertex pairs comes at an algorithmic cost, however: \textsc{Correlation Clustering} is equivalent to \textsc{Edge Multicut} in the sense that there are easy reductions from \textsc{Correlation Clustering} to \textsc{Edge Multicut} and vice versa that preserve the cost~$k$ of the optimal solution~\cite{DemaineEFI06}. With the parameterized complexity status of \textsc{Edge Multicut} being a major open problem at the time, this prompted Bodlaender et al.~\cite{BodlaenderFHMPR10} to investigate a combined parameter that takes the structure of the weight-0 edges, called \emph{fuzzy} edges, into account. The idea was to consider the graph with vertex set~$V$ containing only the fuzzy edges and to use the size of a vertex cover~$r$ of this graph as an additional parameter. Intuitively, the size of the vertex cover is a measure of the amount of missing information. Bodlaender et al. showed that \textsc{Correlation Clustering} admits a kernel with~$\Oh(k^2+r)$~vertices~\cite{BodlaenderFHMPR10}. Subsequently, it was shown that \textsc{Edge Multicut} and with it \textsc{Correlation Clustering} are fixed-parameter tractable with respect to~$k$~\cite{BDT11,MR11}, which implied that the additional parameter~$r$ was not necessary to obtain an FPT-algorithm or a problem kernel. 
Maybe as a consequence, the parameterized complexity of \textsc{Correlation Clustering} with respect to~$k$ plus structural parameterizations of the fuzzy edge graph was not further considered, despite the importance of the problem. We revisit this line of research motivated by the following observation: While the relation to \textsc{Edge Multicut} implies FPT-algorithms with respect to~$k$ it does not give \emph{polynomial} kernels for \textsc{Correlation Clustering}, and in fact hardness results for the related \textsc{Multiway Cut} problem~\cite{CyganKPPW14} imply that such kernels presumably do not exist.\footnote{We will show nonexistence of a kernel in a restricted setting, so we refrain from providing a more detailed argument for this fact.} In other words, to obtain polynomial kernels for~$k$ it is necessary to take the structure of the fuzzy edge graph into account.

\subsection{Our Results}
We consider two parameterizations of the fuzzy edge graph that are both upper-bounded by the maximum degree of the fuzzy edge graph. Thus, in contrast to the global parameter vertex cover number~$r$, these parameters measure the amount of missing information only locally. The first parameter is the degeneracy~$d$ of the fuzzy edge graph, which is the smallest number~$d$ such that every subgraph of the fuzzy edge graph contains at least one vertex with degree at most~$d$. Notably,~$d$ is also upper-bounded by the vertex cover number (and even the treewidth) of the fuzzy edge graph, making the parameterization by~$k+d$ smaller than the one by~$k+r$. Our first main result is that we can design a polynomial problem kernel for this parameter.

\begin{restatable}{theorem}{degeneracykernel}
	\label{thm:degeneracy kernel}
	\CC parameterized by~$k+d$, where~$d$ is the fuzzy degeneracy of the input fuzzy graph, admits a kernel with~$ 30k^3 d$ vertices.
\end{restatable}

The second parameter is the closure~$c$ of the fuzzy edge graph. This graph parameter, introduced by Fox et al.~\cite{FRSWW20} is based on the observation that in many graphs, vertices with many common neighbors tend to be adjacent. Formally, a graph is~$c$-closed whenever any pair of vertices with at least~$c$ common neighbors is adjacent and the closure of a graph is the smallest number~$c$ such that the graph is~$c$-closed. This parameter proved useful as secondary parameter for many different problems~\cite{FRSWW20,KMRSS23,KKNS22,KKS20,KKS20a,KKS23,KN21}. The closure~$c$ is not related to the degeneracy of a graph; for example the closure of a complete graph is~$0$ while the degeneracy is~$n-1$. Hence, it provides an independent way of quantifying the local amount of uncertainty. For example if the uncertainty stems from the fact that there is one set of vertices~$S\subseteq V$ for which we could not empirically determine pairwise similarities inside the set, while all other relations are known, then the closure of the fuzzy edge graph would be 0 even if~$S$ is large. Our second main result is a polynomial problem kernel for~$k+c$. 

\begin{restatable}{theorem}{closurekernel}
	\label{thm:closure kernel}
	\CC parameterized by~$k+c$, where~$c$ is the fuzzy closure of the input fuzzy graph, admits a kernel with $30k^2c^2$ vertices.
\end{restatable}

Of course, it is perfectly reasonable to consider similar parameterizations for the nonfuzzy vertex pairs, for example the degeneracy of the graph consisting of the edges or the structure of the nonedges. Here, we obtain the following negative result which says that we cannot hope for a polynomial kernel even if the degeneracy and the closure of the graph consisting of the real edges and the degree of the graph consisting of the nonedges are both constant.
This result holds even if all real edges and nonedges have unit weight.

\begin{restatable}{theorem}{kernelhardness}
	\label{thm:kernel hardness for real graph}
	\CC is \NP-hard and does not admit a polynomial kernel for the solution size~$k$, unless \bth{}, even if the nonedges form a matching and the graph consisting of the real edges has degeneracy 2 and closure 2.
\end{restatable}

We also show that \CC remains \NP-hard even if the graph consisting of the real edges is a binary tree or a star. 

\begin{theorem}
	\label{thm:hardness for binary trees and stars}
	\CC is \NP-hard even if the graph consisting of the real edges is a binary tree or a star.
\end{theorem}

\section{Preliminaries}
\subparagraph{Graph Notation.}
We consider undirected graphs.
For a graph $G = (V, E)$ with vertex set $V$ and edge set $E$, we denote $n \coloneqq |V|$ and~$m\coloneqq |E|$. 
For a vertex $v$, we let~$N_G (v) \coloneqq \{u \mid \{u,v\} \in E \}$  and $N_G [v] \coloneqq N_G(v) \cup \{v\}$ denote the open and closed neighborhood of~$v$, respectively.
The degree of $v$ is $\deg(v) \coloneqq |N_G(v)|$. 
Analogously, for a set $X \subseteq V$, we define $N_G[X] = \bigcup_{x\in X} N_G[x]$ and $N_G(X) = N_G [X] \setminus X$. 
Moreover, we define $N^{\cap}_G(X) \coloneqq \bigcap_{x \in X} N_G(x)$ and $N^{\cap}_G[X] \coloneqq N^{\cap}_G(X) \cup X$.
We omit subscripts when there is no ambiguity.
The \emph{subgraph of $G$ induced by} $U \subseteq V$ is the graph $G[U] = (U, E_U)$, with $E_U = \{\{x,y\} \in E \mid x,y \in U\}$. 
Given a vertex~$x$, we denote the graph $G[V \setminus \{x\}]$ by $G-x$. 
A graph is \emph{complete} if every pair of vertices are adjacent. A complete subgraph is called a \emph{clique}. 
If a vertex set~$K$ induces a clique, we also say that~$K$ is a clique. 
A clique~$K$ is a \emph{critical clique} if~$N^{\cap}_G[K] = N_G[K]$ and~$K$ is inclusion-wise maximal under this property.
A \emph{vertex cover} is a vertex subset $S \subseteq V$ such that every edge has at least one endpoint in $S$. 

\subparagraph{Fuzzy Graphs.}
A fuzzy graph~$G=(V,E^+,E^0)$ consists of a vertex set~$V$, a set of (\emph{real}) edges~$E^+$ and a set of \emph{fuzzy} edges~$E^0$ with~$E^+\cap E^0=\emptyset$. 
We call a pair of vertices~$\{u,v\}\in {V\choose 2}$ a \emph{nonedge}, if it is neither a real nor a fuzzy edge, and denote the set of nonedges with~$E^- \coloneqq {V \choose 2} \setminus (E^+ \cup E^0)$.
Similar to standard graphs, we define for a vertex~$v\in V$ the real, fuzzy and nonneighborhood as~$N^+(v) \coloneqq \{u \mid \{u,v\} \in E^+ \} $, $N^0(v) \coloneqq \{u \mid \{u,v\} \in E^0 \} $ and $N^-(v) \coloneqq \{u \mid \{u,v\} \in E^- \} $, respectively, and also~$N^+[v] \coloneqq N^+(v) \cup \{v\}$ as well as $N^+[X] = \bigcup_{x\in X} N^+[x]$ for a set~$X\subseteq V$.
We define $G^{(+)}\coloneqq (V,E^+)$, $G^{(0)}\coloneqq (V,E^0)$ and~$G^{(-)}\coloneqq (V,E^-)$, that is, $G^{(+)}$ contains all vertices of~$G$ and its edges are the real edges of~$G$, and similarly~$G^{(0)}$ and~$G^{(-)}$ contain only the fuzzy/nonedges of~$G$ as edges. 
For vertex subsets~$A\subseteq V$ and~$B\subseteq V$ we define~$E^+(A,B)$ as the set of real edges with one endpoint in~$A$ and the other endpoint in~$B$ and analogously~$E^-(A,B)$ and~$E^0(A,B)$.
If one of the sets, for example~$A$, consists of a single vertex~$v$, we also write~$E^+(v,B)$, $E^-(v,B)$ and $E^0(v,B)$.
For a vertex subset~$C\subseteq V$ we write~$E^+(C) \coloneqq E^+(C,C)$, $E^0(C) \coloneqq E^0(C,C)$ and $E^-(C) \coloneqq E^-(C,C)$ for the edges, fuzzy edges and nonedges with both endpoints in~$C$, respectively.
We call a vertex set~$K$ a \emph{fuzzy clique}, if~$K$ is a clique in~$G^{(0)}$.
When speaking of (critical) cliques, neighbors or connected components in a fuzzy graph, we consider the graph~$G^{(+)}$ unless mentioned otherwise. 

\subparagraph{Problem Definition and Basic Observations.}
A clustering is a partition of the vertices of~$V$ and we call each set of the partition a cluster.
For a given clustering~$\mc$ we denote with~$\mc(v)$ the cluster that contains the vertex~$v$.
We call a vertex pair~$\{u,v\}$ a positive mistake if~$\{u,v\} \in E^+ $ and~$u \notin \mc(v)$ and a negative mistake if~$\{u,v\} \in E^-$ and~$u \in \mc(v)$.
For a weight function~$\omega: \binom{V}{2}\to \Z$ and set of vertex pairs~$F$, we let~$\omega(F):=\sum_{e\in F} \omega(e)$ denote the total weight of the vertex pairs in~$F$. 
Note that a vertex pair~$\{u,v\}$ has weight~$0$ if and only if~$\{u,v\}$ is a fuzzy edge.
The \emph{cost} of a clustering~$\mc$ for a weight function~$\omega$ is given by the total weights of the positive and negative mistakes, that is $\cost(\mc) = \sum_{C\in \mc} \omega(E^-(C)) + \sum_{C_1,C_2 \in \mc} \omega(E^+(C_1,C_2)).$

With this we can define our main problem.

	\prob{\CC}{A fuzzy graph~$(V,E^+,E^0)$, a weight function~$\omega:\binom{V}{2}\to \Z$ and a nonnegative integer $k$.}{Is there a clustering~$\mc$ of~$V$ such that~$\cost(\mc)\leq k$?}
	
We call a clustering~$\mc$ a \emph{solution} if~$\cost(\mc)\leq k$ and \emph{optimal} if~$\cost(\mc) \leq \cost(\mc')$ for all clusterings~$\mc'$. We will also take the alternative graph modification view that is made by \textsc{Cluster Editing}: obtain a clustering by deleting real edges~$\{u,v\}$ or inserting nonedges~$\{u,v\}$ at the cost of~$\omega(u,v)$. 
For some of our reduction rules we change a fuzzy edge into a real edge or a nonedge. 
In both cases, we set the new weight of the real edge/nonedge to 1 unless stated otherwise.
A fuzzy graph~$G$ has a clustering with cost~$0$ if and only if all vertices in the connected components of~$G^{(+)}$ are in~$G$ connected via a real or a fuzzy edge; otherwise, at least one real edge or nonedge of~$G$ must be modified. 

If two vertices~$u$ and $v$ have more than~$k$ common neighbors, then separating~$u$ and~$v$ needs more than~$k$ edge deletions.
This implies the following.

\begin{observation}
	\label{obs:many common neighbors}
	Let $u$ and $v$ be vertices with~$|N^+(u)\cap N^+(v)|>k$. Then,~$u$ and~$v$ are contained in the same cluster in any solution.
\end{observation}

\begin{corollary}
	\label{corr:critical cliques}
	Let~$K$ be a clique of size~$k+2$. If~$\mc$ is a clustering with~$\cost(\mc) \le k$, then all vertices of~$N_{G^{(+)}}^{\cap}[K]$ are in the same cluster of~$\mc$.
\end{corollary}
The next observation considers the opposite case, when neighborhoods are too different.
\begin{observation}
	\label{obs:many uncommon neighbors}
	Let $u$ and $v$ be vertices with~$|N^+(u)\cap N^-(v)|>k$. Then,~$u$ and~$v$ are not contained in the same cluster in any solution.
\end{observation}
If there is a cluster that is not connected by real edges, then we could split it up into several clusters without increasing the cost of the clustering, leading to the following observation.

\begin{observation}
	\label{obs:cluster connected by real edges}
	There is an optimal clustering~$\mc$ such that for each cluster~$C \in \mc$ all vertices in~$C$ are connected in~$G^{(+)}$.
\end{observation}

\subparagraph{Parameterized Complexity.}
A parameterized problem is a language $L \subseteq \Sigma^* \times \N$, where $\Sigma$ is a finite alphabet. 
For an instance $(I, k) \in  \Sigma^* \times \N$, we refer to $k$ as the parameter.
An algorithm for a parameterized problem~$ L $ is an~FPT-algorithm if there is a computable function $ f $ such that for every instance $ (I,k) $ the algorithm decides in~$ f(k) \cdot |I|^{\Oh(1)}$ time whether $ (I,k) $ is a yes-instance of $ L $.
A \emph{kernel} for~$L$ is an algorithm that computes for each instance~$(I,k)$ in polynomial time an equivalent instance~$(I',k')$ with~$(|I'|+ k') \leq g(k)$ for a computable function $g$; we call~$g$ the \emph{size} of the kernel.
A kernel of size~$g$ is a \emph{polynomial kernel} if~$g(k) \in k^{\Oh(1)}$.
A parameterized problem has an FPT-algorithm if and only if it admits a (not necessarily polynomial) problem kernel.

A \emph{polynomial parameter transformation} is a reduction from a parameterized problem~$A$ to a parameterized problem~$B$ that transforms each instance~$(I,k)$ of~$A$ in polynomial time into an equivalent instance~$(I',k')$ of~$B$ with~$k' \in k^{\Oh(1)}$.
If~$k' = k$, we say the polynomial parameter transformation preserves the parameter.

Moreover, if there is a polynomial parameter transformation from~$A$ to~$B$, where the unpara\-meterized versions of~$A$ and~$B$ are~\NP-hard and in~\NP, respectively, then~$A$ admits a polynomial kernel if~$B$ admits a polynomial kernel.

For more details about the relevant concepts of parameterized complexity, such as kernelization and kernelization lower bounds, we refer the reader to the standard literature~\cite{DF13,CFK+15}.

\section{Fuzzy Degeneracy Kernel}
In the following, we provide a kernel for parameterization by $k$ plus the degeneracy of~$G^{(0)}$, the fuzzy edge graph of~$G=(V,E^+,E^0)$. We will make use of the following definition of degeneracy via vertex orderings.
For an ordering $\sigma=(v_1,\dots,v_t)$ of a vertex set~$C\subseteq V$ with~$|C|=t$, let $D_{\sigma}[i,j]$ with $i\le j$ denote the set of vertices $\{v_i,\dots,v_j\}$.
Moreover, let $\deg_{\sigma}(v_i) = |N(v_i) \cap D_{\sigma}[i,n]|$ denote the \emph{ordering degree} of~$v_i$ with regards to the ordering~$\sigma$.

A graph~$G=(V,E)$ has \emph{degeneracy}~$d$, if~$d$ is the smallest number for which there is an ordering~$\sigma = (v_1,\dots,v_n)$ of~$V$, such that each vertex $v_i$ has at most~$d$ neighbors~$v_j$ with~$j>i$, that is,~$\deg_{\sigma}(v_i)\le d$.
We call such an ordering a \emph{degeneracy ordering} for~$G$. For a fuzzy graph~$G$, we say that a degeneracy ordering of~$G^{(0)}$ is a \emph{fuzzy degeneracy ordering} of~$G$ and let~$d$ denote the \emph{fuzzy degeneracy} of~$G$. 
For a vertex subset~$C\subseteq V$, let~$\sigma(C)$ be the subsequence of a fuzzy degeneracy ordering~$\sigma$ that contains exactly the vertices of~$C$.
We call such an ordering a fuzzy degeneracy ordering for~$C$.

With these definitions at hand, we may now proceed to describe the kernelization. The first rule is the known rule that removes isolated clusters that do not cause any costs.

\begin{rrule}
	\label{rrule:remove isolated cliques}
	If~$G$ contains a connected component with vertices $C$ such that each pair of vertices~$u,v \in C$ forms a real edge or a fuzzy edge, then remove~$C$ from~$G$. 
\end{rrule}

The following reduction rule makes vertices with many common neighbors adjacent. It corresponds to a classic reduction rule for \textsc{Cluster Editing}~\cite{GGHN05} and its correctness is directly implied by \Cref{obs:many common neighbors}.

\begin{rrule}
	\label{rrule:many common neighbors}
	Let $\{u,v\}$ be a nonedge or fuzzy edge in $G$. If $|N^+(u)\cap N^+(v)| > k$, make $\{u,v\}$ a real edge and set~$k\coloneq k-
        \omega(u,v) $.
\end{rrule}
This rule is already sufficient to show that large clusters must contain large real cliques.
\begin{lemma}
	\label{lem:large cluster contains large clique}
	Let~$G$ be reduced with respect to \Cref{rrule:many common neighbors} and let~$\mc$ be an optimal clustering.
	Let $C\in \mc$ be a cluster of size~$t >2k + 2d$ and let~$\sigma(C) = (v_1,v_2,\dots,v_t)$ be a fuzzy degeneracy ordering of $C$. Then $D_{\sigma(C)}[1,t-2d-2k-1]$ induces a real clique.
\end{lemma}

\begin{proof}
	Each vertex in $X \coloneqq D_{\sigma(C)}[1,t-2d-2k-1]$ has at most $d$ fuzzy neighbors in the last $2k+2d+1$ vertices of the fuzzy degeneracy ordering, that is $Y\coloneq D_{\sigma(C)}[t-2d-2k,t]$.
	
	Let $u$ and~$v$ be arbitrary vertices of~$X$. 
	Then~$u$ and~$v$ have in total at most $2d$ fuzzy neighbors in $Y$ and at most $k$ nonneighbors in $Y$. 	
	Since~$|Y|=2k+2d+1$, $Y$ contains at least $k+1$ common real neighbors of $u$ and $v$. Thus, $u$ and $v$ are neighbors in~$G$ by \Cref{rrule:many common neighbors}. 
	\end{proof}
Finding such large cliques with the aim of shrinking them seems hard, but as we will see, it is sufficient to shrink critical cliques. This is done by the next rule.
\begin{rrule}
	\label{rrule:redirect fuzzy edges}
	Let~$K$ be a critical real clique of size at least $k+3$ and let $v\in K$ be the vertex in~$K$ that appears first in some fuzzy degeneracy ordering~$\sigma$ of~$G$.
	Let $F$ be the set of fuzzy neighbors of $K$ and  
	for each $u\in F$, let $v_u$ be the first fuzzy neighbor of $u$ in $K$ according to~$\sigma$. For each $u\in F$ with $v_u \neq v$ make $\{u,v_u\}$ a nonedge with~$\omega(u,v_u) \coloneqq \omega(u,v)$ and make $\{u,v\}$ a fuzzy edge; then delete $v$ from $G$.
\end{rrule}

\begin{figure}[t]
	\begin{flushleft}
		\begin{tikzpicture}[xscale=0.6,yscale=0.4]
			\tikzstyle{real} = [-]
			\tikzstyle{fuzzy} = [dashed]
			\tikzstyle{non} = [dotted]
			\node[label=below:$v\text{=}v_1$](v1) at (0, 0) [shape = circle, draw, fill=black, scale=0.07ex]{};
			\node[label=below:$v_2$](v2) at (2, 0) [shape = circle, draw, fill=black, scale=0.07ex]{};
			\node[label=below:$v_3$](v3) at (4, 0) [shape = circle, draw, fill=black, scale=0.07ex]{};
			\node[label=below:$v_4$](v4) at (6, 0) [shape = circle, draw, fill=black, scale=0.07ex]{};
			\node[label=below:$v_5$](v5) at (8, 0) [shape = circle, draw, fill=black, scale=0.07ex]{};

			\node[](u1) at (1, 4) [shape = circle, draw, fill=black, scale=0.07ex]{};
			\node[](u2) at (3, 4) [shape = circle, draw, fill=black, scale=0.07ex]{};
			\node[](u3) at (5, 4) [shape = circle, draw, fill=black, scale=0.07ex]{};
			\node[](u4) at (7, 4) [shape = circle, draw, fill=black, scale=0.07ex]{};

			\path [non] (u1) edge (v1);
			\path [fuzzy,line width=0.2mm] (u1) edge (v2);
			\path [non] (u1) edge (v3);
			\path [non] (u1) edge (v4);
			\path [non] (u1) edge (v5);
			\path [fuzzy,line width=0.2mm] (u2) edge (v1);
			\path [non] (u2) edge (v2);
			\path [fuzzy,line width=0.2mm] (u2) edge (v3);
			\path [fuzzy,line width=0.2mm] (u2) edge (v4);
			\path [non] (u2) edge (v5);
			\path [fuzzy,line width=0.2mm] (u3) edge (v1);
			\path [non] (u3) edge (v2);
			\path [non] (u3) edge (v3);
			\path [non] (u3) edge (v4);
			\path [fuzzy,line width=0.2mm] (u3) edge (v5);
			\path [non] (u4) edge (v1);
			\path [non] (u4) edge (v2);
			\path [fuzzy,line width=0.2mm] (u4) edge (v3);
			\path [non] (u4) edge (v4);
			\path [fuzzy,line width=0.2mm] (u4) edge (v5);

			\path[draw,decorate,decoration={snake,amplitude=0.4mm},-stealth] (9,2) -- (11,2);

			\node[label=below:$v\text{=}v_1$](v12) at (12, 0) [shape = cross out, draw, color=red, scale=0.25ex]{};
			\node[label=below:$v_2$](v22) at (14, 0) [shape = circle, draw, fill=black, scale=0.07ex]{};
			\node[label=below:$v_3$](v32) at (16, 0) [shape = circle, draw, fill=black, scale=0.07ex]{};
			\node[label=below:$v_4$](v42) at (18, 0) [shape = circle, draw, fill=black, scale=0.07ex]{};
			\node[label=below:$v_5$](v52) at (20, 0) [shape = circle, draw, fill=black, scale=0.07ex]{};
			
			\node[](u12) at (13, 4) [shape = circle, draw, fill=black, scale=0.07ex]{};
			\node[](u22) at (15, 4) [shape = circle, draw, fill=black, scale=0.07ex]{};
			\node[](u32) at (17, 4) [shape = circle, draw, fill=black, scale=0.07ex]{};
			\node[](u42) at (19, 4) [shape = circle, draw, fill=black, scale=0.07ex]{};
			
			\path [fuzzy,line width=0.2mm] (u12) edge (v12);
			\path [non] (u12) edge (v22);
			\path [non] (u12) edge (v32);
			\path [non] (u12) edge (v42);
			\path [non] (u12) edge (v52);
			\path [fuzzy,line width=0.2mm] (u22) edge (v12);
			\path [non] (u22) edge (v22);
			\path [fuzzy,line width=0.2mm] (u22) edge (v32);
			\path [fuzzy,line width=0.2mm] (u22) edge (v42);
			\path [non] (u22) edge (v52);
			\path [fuzzy,line width=0.2mm] (u32) edge (v12);
			\path [non] (u32) edge (v22);
			\path [non] (u32) edge (v32);
			\path [non] (u32) edge (v42);
			\path [fuzzy,line width=0.2mm] (u32) edge (v52);
			\path [fuzzy,line width=0.2mm] (u42) edge (v12);
			\path [non] (u42) edge (v22);
			\path [non] (u42) edge (v32);
			\path [non] (u42) edge (v42);
			\path [fuzzy,line width=0.2mm] (u42) edge (v52);
			
			\node[draw,shape=rectangle,
			 minimum height=20,
			 minimum width=180,
			 line width=0.8,
			 color=blue,
			 ](e) at (4,-0.3){};
			 
			 \node[draw,shape=rectangle,
			 minimum height=20,
			 minimum width=120,
			 line width=0.8,
			 color=blue,
			 ](e) at (17,-0.3){};
			
			\node[](K) at (-2,0) {$K$};
			\node[](F) at (-2,4) {$F$};
			
			\node[](K') at (21.5,0) {$K'$};
			
		\end{tikzpicture}

	\end{flushleft}
	\caption{An application of \Cref{rrule:redirect fuzzy edges}.
		Fuzzy edges are dashed, nonedges are dotted. After swapping some fuzzy edges and nonedges,~$v$ has only fuzzy neighbors in~$F$ and can thus be safely deleted.}
	\label{fig:redirect fuzzy edges}
\end{figure}
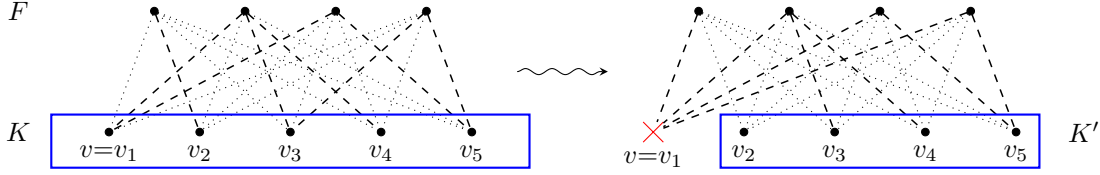

\begin{lemma}
	\Cref{rrule:redirect fuzzy edges} is correct and can be exhaustively applied in~$\Oh(n^3)$ time.
\end{lemma}

\begin{proof}
	Let~$\mc$ be an optimal clustering for~$G$.
	By \Cref{corr:critical cliques}, since $K$ is a critical clique of size at least $k+3$,~$\mc$ contains a cluster~$C_K$ with all vertices of~$N_{G^{(+)}}^{\cap}[K] = N_G^+[K]$.
	Since~$|K|\ge k+3$, the only vertices of~$V \setminus N_G^+[K]$ in~$C_K$ are from~$F$.
	Let~$G'$ be the fuzzy graph obtained after applying \Cref{rrule:redirect fuzzy edges}.
	Note that~$K' \coloneqq K\setminus \{v\}$ is a critical clique of size at least~$k+2$ in~$G'$, so in an optimal clustering for~$G'$ the cluster~$C_{K'}$ containing~$N_{G'}^+[K']$ also additionally contains only vertices of~$F$.
	Thus, it suffices to show that for each~$u\in F$, the cost of adding~$u$ to the cluster containing~$N_{G'}^+[K']$ in~$G'$ is the same as the cost of adding~$u$ to the cluster containing~$N_{G}^+[K]$ in~$G$.
	Let~$u$ be an arbitrary vertex of~$F$. 
	If~$v_u = v$, then~$N^-_{G'}(u) \cap N^+_{G'}[K'] = N^-_G(u) \cap N^+_G[K]$, so the cost of adding~$u$ is the same for~$G$ and~$G'$.
	Otherwise, if~$v_u \neq v$, then $ N^-_{G'}(u) \cap N^+_{G'}[K'] = \left(N^-_G(u) \cap N^+_G[K]\right) \setminus \{v\} \cup \{v_u\} $, so the cost of adding~$u$ is also the same for~$G$ and~$G'$.

	One can compute the critical cliques from the modular decomposition in linear time~\cite{CH94,McCS94}. We then consider the critical cliques one by one.
	For each critical clique,~$F$ and~$v_u$ for all vertices~$u$ can be computed in~$\Oh(n)$ time.
	Since each application of \Cref{rrule:redirect fuzzy edges} removes a vertex, the rule can be applied at most~$n$ times.
	Thus, exhaustive application takes $\Oh(n^3)$~time.
\end{proof}
	
	\begin{lemma}
		\label{lem:kernel size degeneracy}
		Let~$(G=(V,E^+,E^0),k)$ be an instance  of \CC that is reduced with respect to Rules~\ref{rrule:remove isolated cliques}, \ref{rrule:many common neighbors}, and \ref{rrule:redirect fuzzy edges}. If~$|V| > 30k^3 d$, then~$(G,k)$ is a no-instance.
	\end{lemma}
	
	\begin{proof}
		Let~$(G,k)$ be a yes-instance and let~$\mc$ be a solution for~$(G,k)$.
		Let $C$ be a cluster of size at least $3k+2d+3$ and let~$\sigma(C) = (v_1,v_2,\dots,v_t)$ be a fuzzy degeneracy ordering of $C$.
		We show that~$|C| \le 14k^2 d$. 		
		Let~$X \coloneqq D_{\sigma(C)}[1,t-2d-2k-1]$ be the vertices of~$C$ except for the last $2k+2d+1$ vertices according to~$\sigma(C)$. 
		According to \Cref{lem:large cluster contains large clique} the vertices of~$X$ form a clique.
		Let $C_K \subseteq C$ be a maximal real clique in $G$ that contains $X$ and potentially parts of $Y \coloneqq C\setminus X$ and let $C_F \coloneqq C\setminus C_K$.
		Note that~$C_F \subseteq Y$ and thus~$|C_F| \leq 2k+2d+1$. 
		We can partition $C_K$ as follows (see also \Cref{fig:cluster structure degeneracy}).
		\begin{itemize}
			\item Let~$C_{KF} \subseteq C_K$ be the vertices in $C_K$ that have neighbors in~$C_F$,
			\item let $\widetilde{C}_K \subseteq C_K \setminus C_{KF}$ be the vertices in $C_K \setminus C_{KF}$ that have neighbors outside of~$C$, and
			\item let~$C^*_K = C_K \setminus (C_{KF} \cup \widetilde{C}_K)$ be the remaining vertices in~$C_K$.
		\end{itemize}
		
		Per definition, vertices in~$C^*_K$ only have fuzzy neighbors and nonneighbors outside of~$C$ and thus $C^*_K$ is a critical clique with~$N^+[C^*_K] = C_K$.
		Observe that $|\widetilde{C}_K| \leq k$ since every real edge with an endpoint outside of $C$ has to be deleted to obtain the cluster~$C$. 		
		
		Moreover, we have $|C_{KF}| \leq (k+1)\cdot (2k+2d+1)$:
		Otherwise, since $|C_F| \leq 2k+2d+1$, per pigeonhole principle there must be some vertex $v\in C_F$ that has at least $k+1$ neighbors in $C_{KF}$.
		Then, since $C_K$ is a clique of size at least~$k+2$, for each $u\in C_K$ there are $k+1$ common neighbors of $u$ and $v$ in $C_{KF}$.
		Thus, according to \Cref{rrule:many common neighbors}, $\{u,v\}\in E^+$ and therefore $C_K \cup \{v\}$ is a clique.
		This is a contradiction to the maximality of~$C_K$.
		
		Finally, we have~$|C^*_K| \leq k+2$ , since~$C^*_K$ is a critical clique and~$G$ is reduced with respect to \Cref{rrule:redirect fuzzy edges}.
		Thus, we can bound the size of~$C$:
		\begin{align*}
			|C| & = |C_F| + |C_{KF}| + |\widetilde{C}_K| + |C^*_K| \\
				& \leq (2k+2d+1) + (k+1) \cdot (2k+2d+1) + k + (k+2) \\
				& = 2k^2 + 2kd + 7k + 4d + 4 \\
				& \leq 15k^2 d. 				
		\end{align*}

		Since~$G$ is reduced with respect to \Cref{rrule:remove isolated cliques}, each cluster of~$\mc$ contains a vertex that is incident to a modified edge in~$G$, so~$\mc$ contains at most~$2k$ clusters, each of which contains at most $15k^2 d$ vertices.
		Thus,~$G$ contains at most~$ 30k^3 d$ vertices.	
	\end{proof}
	
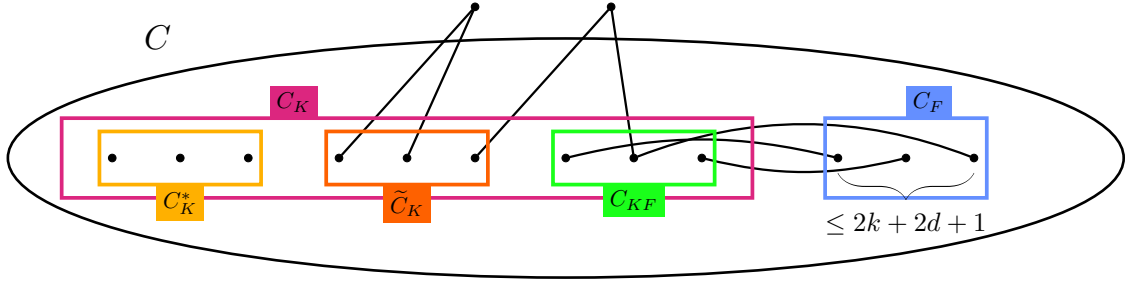
\begin{figure}[t]
	\label{fig:cluster structure degeneracy}

\begin{tikzpicture}[xscale=0.6,yscale=0.4]
	\tikzstyle{real} = [-]
	\tikzstyle{fuzzy} = [dashed]
	\tikzstyle{non} = [dotted]
	\node[](v1) at (0, 0) [shape = circle, draw, fill=black, scale=0.07ex]{};
	\node[](v2) at (1.5, 0) [shape = circle, draw, fill=black, scale=0.07ex]{};
	\node[](v3) at (3, 0) [shape = circle, draw, fill=black, scale=0.07ex]{};
	\node[](v4) at (5, 0) [shape = circle, draw, fill=black, scale=0.07ex]{};
	\node[](v5) at (6.5, 0) [shape = circle, draw, fill=black, scale=0.07ex]{};
	\node[](v6) at (8, 0) [shape = circle, draw, fill=black, scale=0.07ex]{};
	\node[](v7) at (10, 0) [shape = circle, draw, fill=black, scale=0.07ex]{};
	\node[](v8) at (11.5, 0) [shape = circle, draw, fill=black, scale=0.07ex]{};
	\node[](v9) at (13, 0) [shape = circle, draw, fill=black, scale=0.07ex]{};
	\node[](v10) at (16, 0) [shape = circle, draw, fill=black, scale=0.07ex]{};
	\node[](v11) at (17.5, 0) [shape = circle, draw, fill=black, scale=0.07ex]{};
	\node[](v12) at (19, 0) [shape = circle, draw, fill=black, scale=0.07ex]{};
	
	\node[](u1) at (8, 5) [shape = circle, draw, fill=black, scale=0.07ex]{};
	\node[](u2) at (11, 5) [shape = circle, draw, fill=black, scale=0.07ex]{};
	\node[](C) at (1, 4) [shape = circle, scale=0.3ex]{$C$};

	\path [real,line width=0.3mm] (v4) edge (u1);
	\path [real,line width=0.3mm] (v5) edge (u1);
	\path [real,line width=0.3mm] (v6) edge (u2);
	\path [real,line width=0.3mm] (v8) edge (u2);
	
	\path [bend left=20,real,line width=0.3mm] (v7) edge (v10);
	\path [bend left=30,real,line width=0.3mm] (v8) edge (v12);
	\path [bend right=20,real,line width=0.3mm] (v9) edge (v11);

	\node[draw,shape=rectangle,
	minimum height=20,
	line width=1.4,
	color=mycolor5,
	fit= (v1) (v2) (v3) 
	](e) at (1.5,0){};
	
	\node[draw,shape=rectangle,
	minimum height=20,
	line width=1.4,
	color=mycolor4,
	fit= (v4) (v5) (v6) 
	](e) at (6.5,0){};

	\node[draw,shape=rectangle,
	minimum height=20,
	line width=1.4,
	color=mycolor2,
	fit= (v4) (v5) (v6) 
	](e) at (11.5,0){};
	
	\node[draw,shape=rectangle,
	minimum height=30,
	line width=1.4,
	color=mycolor1,
	fit= (v7) (v8) (v9) 
	](e) at (17.5,0){};
	
	\node[draw,shape=ellipse,
	minimum height=90,
	minimum width=420,
	line width=1,
	](e) at (10,0){};

	\node[draw,shape=rectangle,
	minimum height=30,
	minimum width=260,
	line width=1.4,
	color=mycolor3,
	fit= (v1) (v2) (v3) (v4) (v5) (v6) (v7) (v8) (v9) 
	](e) at (6.5,0){};

	\node[](CK) at (4,1.85) [shape = circle, scale=0.2ex]{\colorbox{mycolor3}{$C_{K}$}};
	
	\node[](CF) at (18,1.85) [shape = circle, scale=0.2ex]{\colorbox{mycolor1}{$C_{F}$}};
	
	\node[](C3) at (1.5,-1.5) [shape = circle, scale=0.2ex]{\colorbox{mycolor5}{$C^*_K$}};
	
	\node[](C2) at (6.5,-1.5) [shape = circle, scale=0.2ex]{\colorbox{mycolor4}{$\widetilde{C}_K$}};
	
	\node[](C1) at (11.5,-1.4) [shape = circle, scale=0.2ex]{\colorbox{mycolor2}{$C_{KF}$}};

	\draw [decorate,decoration={calligraphic brace,amplitude=10pt},xshift=0pt,yshift=0pt]
	(19,-0.5) -- (16,-0.5) node [black,midway,xshift=0cm,yshift=-0.67cm] 
	{$\leq 2k +2d+1$};
\end{tikzpicture}

	\caption{Structure of a cluster~$C$ after application of all reduction rules. $C_K$ is a clique and $C_F$ contains a subset of the last~$2k+2d+1$ vertices according to a given fuzzy degeneracy ordering.}

\end{figure}
	
\Cref{lem:kernel size degeneracy} directly implies \Cref{thm:degeneracy kernel}.
	
	\degeneracykernel*

\section{Fuzzy Closure Kernel}
We now exploit the closure of the graph~$G^{(0)}$ induced by the fuzzy edges. Recall that a graph~$H$ is \emph{$c$-closed} for some integer~$c$, if each pair of nonneighbors~$u$ and $v$ has less than~$c$ common neighbors. The~\emph{closure} of a graph~$H$ is the smallest number~$c$ for which~$H$ is~$c$-closed.
We call the closure of~$G^{(0)}$  the \emph{fuzzy closure} of~$G$ and say a fuzzy graph~$G$ is \emph{fuzzy $c$-closed}, if~$G^{(0)}$ is~$c$-closed.
For a given value~$c$, we call a pair of vertices~$u$ and $v$ \emph{bad}, if they have at least~$c$ common fuzzy neighbors, but~$\{u,v\}$ is not a fuzzy edge.
Thus a fuzzy graph~$G$ has fuzzy closure at most~$c$ if and only if~$G$ does not have any bad pair for~$c$.

We present a polynomial kernel for the parameter~$k+c$, where~$c$ is the fuzzy closure of the input fuzzy graph~$G$.
Note that if~$\{u,v\}$ is a fuzzy edge, then making~$\{u,v\}$ a real or nonedge may change the fuzzy closure of the graph.
Hence, the first rule is a special case of~\Cref{rrule:many common neighbors} that is only applied to nonedges.
 \begin{rrule}
 	\label{rrule:many common neighbors closure}
 	Let $\{u,v\}$ be a nonedge in $G$. If $|N^+(u)\cap N^+(v)| > k$, make $\{u,v\}$ a real edge and set~$k\coloneq k-\omega(u,v)$.
 \end{rrule}
 
 We can also state a similar rule for real edges that is directly implied by~\Cref{obs:many uncommon neighbors}.
 
  \begin{rrule}
 	\label{rrule:many one sided neighbors}
 	Let $\{u,v\}$ be a real edge in $G$. If $|N^+(u)\cap N^-(v)| > k$, make $\{u,v\}$ a nonedge and set~$k\coloneq k-\omega(u,v)$.
 \end{rrule}
 
 If~$u$ and $v$ have less than~$c$ common fuzzy neighbors, then they can not be a bad pair and we can apply a similar reduction rule for a fuzzy edge~$\{u,v\}$.
  
  \begin{rrule}
 	\label{rrule:many common neighbors fuzzy}
 	Let $\{u,v\}$ be a fuzzy edge in $G$ such that~$|N^0(u) \cap N^0(v)| < c$.
        \begin{itemize}
        \item If $|N^+(u)\cap N^+(v)| > k$, make $\{u,v\}$ a real edge.
        \item If $|N^+(u)\cap N^-(v)| > k$, make $\{u,v\}$ a nonedge.
        \end{itemize}
 \end{rrule}
 \begin{lemma}
 	\Cref{rrule:many common neighbors fuzzy} is correct and can be exhaustively applied in~$\Oh(n^3)$ time.
 	If~$G$ is fuzzy~$c$-closed, then the resulting graph is also fuzzy~$c$-closed.
 \end{lemma} 
 \begin{proof}
 	According to \Cref{obs:many common neighbors}, if $|N^+(u)\cap N^+(v)| > k$, then  any solution puts~$u$ and~$v$ in the same cluster, so we can safely make~$\{u,v\}$ a real edge.
 	Let~$G'$ be the fuzzy graph obtained by the application of the rule.
 	Since~$|N^0(u) \cap N^0(v)| < c$, the vertices $\{u,v\}$ are not a bad pair in~$G'$.
 	For every other pair of vertices~$\{x,y\} \neq \{u,v\}$, the number of common fuzzy neighbors does not increase by making~$\{u,v\}$ a real edge, so~$G'$ is still fuzzy $c$-closed.
 	We can exhaustively apply the rule by checking the common fuzzy/real neighborhood for each pair of vertices and updating the edges accordingly in~$\Oh(n^3)$ time. The correctness of the case $|N^+(u)\cap N^-(v)| > k$ can be shown analogously. 
 \end{proof}
 
 We will now show that we can modify the instance so that afterwards each vertex has a bounded number of real neighbors or a bounded number of fuzzy neighbors, which we then use to bound the size of clusters in a solution.
 
 \begin{lemma}
 	 \label{lem:not many real and fuzzy edges}
 	Let~$v$ be a vertex with more than~$2kc$ real neighbors and at least~$c$ fuzzy neighbors. 
 	Then~$v$ has a fuzzy neighbor~$u$ with $|N^+(v) \cap N^+(u)| > k$ or $|N^+(v) \cap N^-(u)| > k$.
 \end{lemma}
 
 \begin{proof}
 	Let~$X\subseteq N^0(v)$ be a set of $c$ fuzzy neighbors of~$v$.
 	Each real neighbor~$w \in N^+(v)$ of~$v$ has at least one real or nonneighbor in~$X$, since otherwise~$G$ would not be fuzzy $c$-closed. 
 	Since~$|N^+(v)| > 2kc$ and~$|X|=c $, per pigeonhole principle there must be a vertex~$u\in X$ with $|N^+(v) \cap N^+(u)| > k$ or $|N^+(v) \cap N^-(u)| > k$.
 \end{proof}
 
 \begin{rrule}
 	\label{rrule:not many real and fuzzy edges}
	Let~$v$ be a vertex with more than~$2kc$ real neighbors and at least~$c$ fuzzy neighbors. 
	Then, find a vertex~$u\in N^0(v)$ such that~$|N^+(v) \cap N^+(u)| > k$ and make~$\{u,v\}$ a real edge or find a vertex~$u\in N^0(v)$ such that~$|N^+(v) \cap N^-(u)| > k$ and make~$\{u,v\}$ a nonedge.
 \end{rrule}
 
 \begin{lemma}
 	\Cref{rrule:not many real and fuzzy edges} is correct and can be exhaustively applied in $\Oh(n^3)$ time.
 \end{lemma}
 
 \begin{proof}
 	In the first case,~$v$ and~$u$ have to be in the same cluster, since otherwise we would have to delete more than~$k$ edges to separate them, thus we can safely make the fuzzy edge~$\{u,v\}$ a real edge. Analogously, in the second case~$v$ and~$u$ cannot be in the same cluster with at most~$k$ edge modifications, so we can safely make~$\{u,v\}$ a nonedge.
 	
 	Checking whether a vertex has sufficiently many real and fuzzy neighbors as well as checking for each fuzzy neighbor the real/nonneighborhood can be done in~$\Oh(n^3)$ time.
 \end{proof}
 
  We now show that after exhaustively applying~\Cref{rrule:not many real and fuzzy edges} we can assume that each vertex in~$G$ has at most~$2kc$ real neighbors or less than~$c$ fuzzy neighbors, and the resulting fuzzy graph is still fuzzy $c$-closed.
 \begin{lemma}
 	Let~$G'$ be the resulting fuzzy graph after exhaustively applying~\Cref{rrule:not many real and fuzzy edges}.
 	Then, for each vertex~$v\in V$ we have~$|N_{G'}^0(v)| < c$ or~$|N^+_{G'}(v)| \leq 2kc$ and~$G'$ is fuzzy $c$-closed.
 \end{lemma}
 
 \begin{proof}
 	Assume towards a contradiction that there is a vertex~$v\in V$ in~$G'$ with more than~$2kc$ real neighbors and at least~$c$ fuzzy neighbors.
 	Then, according to \Cref{lem:not many real and fuzzy edges} there is a vertex~$u\in N^0(v)$ with $|N^+(v) \cap N^+(u)| > k$ or $|N^+(v) \cap N^-(u)| > k$ and thus \Cref{rrule:not many real and fuzzy edges} can be applied, leading to a contradiction.

 	Thus, it remains to show that~$G'$ is also fuzzy~$c$-closed. 	
 	Let~$v$ be a vertex in~$G$ with more than~$2kc$ real neighbors and at least~$c$ fuzzy neighbors to which~\Cref{rrule:not many real and fuzzy edges} is applied. 
 	Note that applying~\Cref{rrule:not many real and fuzzy edges} once can only create a bad pair~$\{u,v\}$ that is formed by a fuzzy edge~$\{u,v\}$ that was made into a real edge or nonedge, since~$G$ is~$c$-closed and for each other pair of vertices the common fuzzy neighborhood either stays the same or gets smaller.
 	As long as~$v$ still has at least~$c$ fuzzy neighbors,~\Cref{rrule:not many real and fuzzy edges} can be applied again, where after each application a fuzzy neighbor of~$v$ is then made a real or a nonneighbor, so the number of fuzzy neighbors of~$v$ decreases with each application of~\Cref{rrule:not many real and fuzzy edges}. 
 	Hence, after exhaustively applying~\Cref{rrule:not many real and fuzzy edges},~$v$ has less than~$c$ fuzzy neighbors and can thus not be in a bad pair anymore.  	 	
 \end{proof} 

 In the following we assume~$G$ is reduced with regards to~\Cref{rrule:not many real and fuzzy edges} and define~$V_+$ to be the set of vertices with less than~$c$ fuzzy neighbors and~$V_0 \coloneqq V \setminus V_+$.
 Note that each vertex in~$V_0$ has at most~$2kc$ real neighbors.
 Next, we show that if a cluster is large, then we can partition it into a real clique and a fuzzy clique.

 \begin{lemma}
	\label{lem:clique and fuzzy clique in pure cluster}
	Let~$\mc$ be a solution for~$(G,k)$ and let~$C\in \mc$ be a cluster. 
	If~$|C| > 4kc + k + c +2$, then~$C\cap V_0$ is a clique in~$G^{(0)}$ and~$C\cap V_+$ is a clique in~$G^{(+)}$.
\end{lemma}

 \begin{proof}
 	Let~$u$ and $v$ be vertices of $C \cap V_0$. 
 	By definition of~$V_0$, each of~$u$ and~$v$ has at most~$2kc$ real neighbors in~$C$.
 	Moreover, since~$C$ is a cluster of a solution,~$u$ and~$v$ have in total at most~$k$ nonneighbors in~$C$.
 	Thus, more than~$c$ vertices of~$C$ are common fuzzy neighbors of~$u$ and~$v$. 
 	Since~$G$ has fuzzy closure~$c$, the vertices~$u$ and~$v$ must form a fuzzy edge.

	Now let~$u$ and~$v$ be arbitrary vertices of $C \cap V_+$.
 	Per definition of~$V_+$, each of~$u$ and~$v$ has less than~$c$ fuzzy neighbors in~$C$.
 	Moreover, since~$C$ is a cluster,~$u$ and~$v$ have in total at most~$k$ nonneighbors in~$C$.
 	Since~$|C| > 4kc + k + c +2 > 2c+2k+2$, more than~$k$ vertices of~$C$ are common real neighbors of~$u$ and~$v$.
 	Thus, according to \Cref{rrule:many common neighbors closure} and  \Cref{rrule:many common neighbors fuzzy}, the vertices~$u$ and~$v$ must form a real edge.
 \end{proof}
 
 We can bound the size of a clique in $G$ with the following reduction rule.
 \begin{rrule}
 	\label{rrule:bounded cliques c-closure}
 	Let~$K$ be a real clique of size at least~$3kc+2$ and let~$v$ be an arbitrary fixed vertex of~$K$. Then for each vertex~$u \in V \setminus K$ with~$\{u,v\} \notin E^0$ do the following:
 	\begin{itemize}
 		\item Let~$w \in K$ be a vertex such that~$\{u,w\}\in E^0$, if such a vertex exists.
 		\item If $\{u,v\}$ is a real edge, make $\{u,w\}$ a real edge with~$\omega(u,w) \coloneqq \omega(u,v)$.
 		\item If $\{u,v\}$ is a nonedge, make $\{u,w\}$ a nonedge with~$\omega(u,w) \coloneqq \omega(u,v)$.
 	\end{itemize} 
 	Afterwards, delete~$v$.
 \end{rrule}
 
 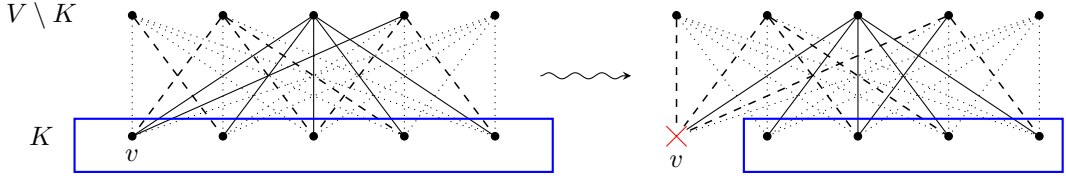
\begin{figure}[t]
 	\begin{flushleft}
 		\begin{tikzpicture}[xscale=0.6,yscale=0.4]
 			\tikzstyle{real} = [-]
 			\tikzstyle{fuzzy} = [dashed]
 			\tikzstyle{non} = [dotted]
 			\node[label=below:$v$](v1) at (0, 0) [shape = circle, draw, fill=black, scale=0.07ex]{};
 			\node[](v2) at (2, 0) [shape = circle, draw, fill=black, scale=0.07ex]{};
 			\node[](v3) at (4, 0) [shape = circle, draw, fill=black, scale=0.07ex]{};
 			\node[](v4) at (6, 0) [shape = circle, draw, fill=black, scale=0.07ex]{};
 			\node[](v5) at (8, 0) [shape = circle, draw, fill=black, scale=0.07ex]{};
 			
 			\node[](u1) at (0, 4) [shape = circle, draw, fill=black, scale=0.07ex]{};
 			\node[](u2) at (2, 4) [shape = circle, draw, fill=black, scale=0.07ex]{};
 			\node[](u3) at (4, 4) [shape = circle, draw, fill=black, scale=0.07ex]{};
 			\node[](u4) at (6, 4) [shape = circle, draw, fill=black, scale=0.07ex]{};
 			\node[](u5) at (8, 4) [shape = circle, draw, fill=black, scale=0.07ex]{};
 			
 			\path [non] (u1) edge (v1);
 			\path [fuzzy,line width=0.2mm] (u1) edge (v2);
 			\path [non] (u1) edge (v3);
 			\path [non] (u1) edge (v4);
 			\path [non] (u1) edge (v5);
 			\path [fuzzy,line width=0.2mm] (u2) edge (v1);
 			\path [non] (u2) edge (v2);
 			\path [fuzzy,line width=0.2mm] (u2) edge (v3);
 			\path [fuzzy,line width=0.2mm] (u2) edge (v4);
 			\path [non] (u2) edge (v5);
 			\path [real] (u3) edge (v1);
 			\path [real] (u3) edge (v2);
 			\path [real] (u3) edge (v3);
 			\path [real] (u3) edge (v4);
 			\path [real] (u3) edge (v5);
 			\path [real] (u4) edge (v1);
 			\path [non] (u4) edge (v2);
 			\path [fuzzy,line width=0.2mm] (u4) edge (v3);
 			\path [non] (u4) edge (v4);
 			\path [fuzzy,line width=0.2mm] (u4) edge (v5);
 			\path [non] (u5) edge (v1);
 			\path [non] (u5) edge (v2);
 			\path [non] (u5) edge (v3);
 			\path [non] (u5) edge (v4);
 			\path [non] (u5) edge (v5);

 			\path[draw,decorate,decoration={snake,amplitude=0.4mm},-stealth] (9,2) -- (11,2);

 			\node[label=below:$v$](v12) at (12, 0) [shape = cross out, draw, color=red, scale=0.25ex]{};
 			\node[](v22) at (14, 0) [shape = circle, draw, fill=black, scale=0.07ex]{};
 			\node[](v32) at (16, 0) [shape = circle, draw, fill=black, scale=0.07ex]{};
 			\node[](v42) at (18, 0) [shape = circle, draw, fill=black, scale=0.07ex]{};
 			\node[](v52) at (20, 0) [shape = circle, draw, fill=black, scale=0.07ex]{};
 			
 			\node[](u12) at (12, 4) [shape = circle, draw, fill=black, scale=0.07ex]{};
 			\node[](u22) at (14, 4) [shape = circle, draw, fill=black, scale=0.07ex]{};
 			\node[](u32) at (16, 4) [shape = circle, draw, fill=black, scale=0.07ex]{};
 			\node[](u42) at (18, 4) [shape = circle, draw, fill=black, scale=0.07ex]{};
 			\node[](u52) at (20, 4) [shape = circle, draw, fill=black, scale=0.07ex]{};
 			
 			\path [fuzzy,line width=0.2mm] (u12) edge (v12);
 			\path [non] (u12) edge (v22);
 			\path [non] (u12) edge (v32);
 			\path [non] (u12) edge (v42);
 			\path [non] (u12) edge (v52);
 			\path [fuzzy,line width=0.2mm] (u22) edge (v12);
 			\path [non] (u22) edge (v22);
 			\path [fuzzy,line width=0.2mm] (u22) edge (v32);
 			\path [fuzzy,line width=0.2mm] (u22) edge (v42);
 			\path [non] (u22) edge (v52);
 			\path [real] (u32) edge (v12);
 			\path [real] (u32) edge (v22);
 			\path [real] (u32) edge (v32);
 			\path [real] (u32) edge (v42);
 			\path [real] (u32) edge (v52);
 			\path [fuzzy,line width=0.2mm] (u42) edge (v12);
 			\path [non] (u42) edge (v22);
 			\path [real] (u42) edge (v32);
 			\path [non] (u42) edge (v42);
 			\path [fuzzy,line width=0.2mm] (u42) edge (v52);
 			\path [non] (u52) edge (v12);
 			\path [non] (u52) edge (v22);
 			\path [non] (u52) edge (v32);
 			\path [non] (u52) edge (v42);
 			\path [non] (u52) edge (v52);
 				
 			\node[draw,shape=rectangle,
 			minimum height=20,
 			minimum width=180,
 			line width=0.8,
 			color=blue,
 			](e) at (4,-0.3){};
 			
 			\node[draw,shape=rectangle,
 			minimum height=20,
 			minimum width=120,
 			line width=0.8,
 			color=blue,
 			](e) at (17,-0.3){};
 			
 			\node[](K) at (-2,0) {$K$};
 			\node[](F) at (-2,4) {$V \setminus K$};
 		\end{tikzpicture}

 	\end{flushleft}
 	\caption{An application of \Cref{rrule:bounded cliques c-closure}.
 		Fuzzy edges are dashed, nonedges are dotted. After swapping some real edges and nonedges with fuzzy edges, each real neighbor/nonneighbor of~$v$ in~$V\setminus K$ is also a real neighbor/nonneighbor of every other vertex in~$K$ and thus~$v$ can be safely deleted.}
 	\label{fig:redirect fuzzy edges closure}
 \end{figure} 
 
 \begin{lemma}
 	\Cref{rrule:bounded cliques c-closure} is correct and can be exhaustively applied in~$\Oh(n^4)$ time.
 	 If~$G$ is fuzzy~$c$-closed, then the resulting graph is also fuzzy~$c$-closed. 
 \end{lemma}
  \begin{proof}
  	Let~$G'$ be the fuzzy graph after applying \Cref{rrule:bounded cliques c-closure} with vertex set~$V' \coloneqq V \setminus \{v\}$.
  	We show that~$(G,k)$ is a yes-instance if and only if~$(G',k)$ is a yes-instance.
  	First, suppose that~$(G,k)$ is a yes-instance and $\mc$ is a clustering for~$G$ with~$\cost(\mc) \leq k$.
 	According to \Cref{corr:critical cliques} all vertices of~$K$ are contained in the same cluster, which we denote by~$C_K$.
 	Now consider the clustering~$\mc' \coloneqq (\mc \setminus \{C_K\}) \cup \{C_K \setminus \{v\}\} $ for~$G'$ that is identical to~$\mc$ except that~$C_K$ is replaced by~$C'_K \coloneqq C_K \setminus\{v\}$.
 	Note that if for each~$u\in V\setminus K$ we have~$\{u,v\} \in E_G^0$, then~$G'$ is an induced subgraph of~$G$ and thus~$\cost(\mc') \leq \cost(\mc) \leq k$.
 	Otherwise, let~$U \subseteq V \setminus K$ be the vertices that are not fuzzy neighbors of~$v$.
 	Consider some fuzzy edge~$\{u,w\}\in E_G^0$ with~$u\in U$, $w\in K$ that becomes a real/nonedge in~$G'$ and contributes to~$\cost(\mc')$. 
 	Then,~$\{u,v\}$ must also be a real/nonedge in~$G$ and equally contribute to~$\cost(\mc)$, so we get~$\cost(\mc') \leq \cost(\mc) \leq k$ again.
 
	Now, suppose that~$(G',k)$ is a yes-instance and $\mc'$ is a clustering for~$G'$ with $\cost(\mc') \leq k$.
	Note that, since~$|K\setminus \{v\}| \geq 3kc+1$, again according to \Cref{corr:critical cliques} all vertices of~$K\setminus \{v\}$ are contained in the same cluster, which we denote by~$C'_K$.
	Consider the clustering~$\mc \coloneqq (\mc' \setminus \{C'_K\}) \cup \{C'_K \cup \{v\}\} $ for~$G'$ that is identical to~$\mc'$ except that~$C'_K$ is replaced by~$C_K \coloneqq C'_K \cup \{v\}$.
 	We have
 	\begin{align*}
 		\cost(\mc) - \cost(\mc') & = \omega(E_G^+(C_K,V\setminus C_K))+\omega(E_G^-(C_K)) - \omega(E_{G'}^+(C'_K,V'\setminus C'_K))-\omega(E_{G'}^-(C'_K)) 		
 	\end{align*}
 	and it remains to show that~$\cost(\mc) - \cost(\mc') \leq 0$.
 	Clearly, the contribution of edges between two vertices of~$K \setminus \{v\}$ is the same for both~$\cost(\mc)$ and~$\cost(\mc')$.
 	The same is true for edges between two vertices of~$V\setminus K$.
 	Now consider an arbitrary vertex~$u \in V\setminus K$.
 	If~$\{u,v\}\in E_G^0$, then $\{u,v\}$ does not contribute to neither~$\cost(\mc)$ nor~$\cost(\mc')$ and for each~$w\in K$ the contribution of~$\{u,w\}$ to~$\cost(\mc)$ and~$\cost(\mc')$ is the same.
 	
 	Thus, now assume that~$\{u,v\}\notin E_G^0$.
 	If $|N_G^0(u) \cap K| \geq 1$, then there is a vertex~$w\in K$ with~$\{u,w\}\in E_G^0$ as stated in \Cref{rrule:bounded cliques c-closure}.
 	If $\{u,v\}$ is a real edge in~$G$, then $\{u,w\}$ becomes a real edge in~$G'$.
 	Thus, if~$\{u,v\}$ contributes to~$\omega(E_G^+(C_K,V\setminus C_K))$, then~$\{u,w\}$ contributes to~$\omega(E_{G'}^+(C'_K,V'\setminus C'_K))$.
 	Similarly, if $\{u,v\}$ is a nonedge in~$G$, then $\{u,w\}$ becomes a nonedge in~$G'$ and if~$\{u,v\}$ contributes to~$\omega(E_G^-(C_K))$, then~$\{u,w\}$ contributes to~$\omega(E_{G'}^-(C'_K))$.
 	For each~$z\in K \setminus\{w\}$ the contribution of~$\{u,z\}$ to~$\cost(\mc)$ and~$\cost(\mc')$ is the same.
 	
 	Otherwise, we either have $|N_G^+(u) \cap K| > k+1$ or~$|N_G^-(u) \cap K| > k+1$. 	
 	If~$|N_G^+(u) \cap K| > k+1$, then~$u$ and each vertex in~$K$ have more than~$k$ common neighbors in~$K$, so according to \Cref{rrule:many common neighbors closure} and \Cref{rrule:many common neighbors fuzzy} for each~$w\in K$ we have that~$\{u,w\}$ is a real edge, which implies $u\in C_K$ and thus the contribution of edges in~$E_G^+(u,K)$ to $\omega(E_G^+(C_K,V\setminus C_K))$ is zero.
 	Similarly, if~$|N_G^-(u) \cap K| > k+1$, then according to \Cref{rrule:many one sided neighbors} and \Cref{rrule:many common neighbors fuzzy} for each~$w\in K$ we have that~$\{u,w\}$ is a nonedge, which implies~$u\notin C_K$ and the contribution of edges in~$E_G^-(u,K)$ to $\omega(E_G^-(C_K))$ is zero.
	All in all we get~$\cost(\mc) - \cost(\mc') \leq 0$.
	 	
	Note that, since~$G$ is fuzzy $c$-closed, applying \Cref{rrule:bounded cliques c-closure} once could potentially only create a bad pair~$\{u,w\}$ that was made a real edge or a nonedge for a vertex~$w\in K$.
	However, since~$K$ is a real clique of size~$3kc+2$, each vertex~$w\in K$ is in~$V_+$, has less than~$c$ fuzzy neighbors and thus cannot be part of a bad pair.
 	
 	In~$\Oh(n^3$) time we can find a clique of size~$3kc+2$ or verify that such a clique does not exist:
 	Let~$K$ be a clique of size~$3kc+2$ and let~$u$ be a vertex in~$K$. 
 	Clearly, only real neighbors of~$u$ with at least~$3kc$ common neighbors with~$u$ can possibly be part of~$K$. 
 	Now, let~$w\in N^+(u)$ with~$|N^+(u) \cap N^+(w)| \geq 3kc$ be such a neighbor of~$u$
 	and assume towards a contradiction that~$w\notin K$ and~$K \cup \{w\}$ is not a clique.
 	Since~$w\in V_+$, we have~$|N^0(w) \cap K| < c$.
 	Moreover, $|N^-(w) \cap K| < k+1$, since otherwise~$\{u,w\}$ would be a nonedge according to \Cref{rrule:many one sided neighbors}.
 	This implies~$|N^+(w) \cap K| > k+1$ and thus~$w$ has more than~$k$ common neighbors with each vertex in~$K$, so~$K\cup \{w\}$ must also be a clique according to \Cref{rrule:many common neighbors closure}.
 	Thus, we can iterate over all vertices and for a fixed~$u\in V$ we check whether~$u$ together with each neighbor~$v$ with at least~$3kc$ common neighbors forms a clique, which can be done in $\Oh(n^3)$ time. 
 	
 	Given a clique of size at least~$3kc+2$, we can then apply \Cref{rrule:bounded cliques c-closure} in~$\Oh(n^2)$ time by iterating through~$V\setminus K$ and~$K$.
 	Since each application of~\Cref{rrule:bounded cliques c-closure} removes a vertex, it can be applied at most~$n$ times and thus also exhaustively applied in~$\Oh(n^4)$ time. 
 \end{proof}

The final reduction rule will help bounding the size of  fuzzy cliques inside  clusters.

 \begin{rrule}
 	\label{rrule:bound fuzzy clique in pure clusters}
 	Let~$K$ be a real clique in~$G$ such that~$X \coloneqq N_{G^{(+)}}^{\cap}(K)\cap \{v\in V \mid N^+(v) = K \}$ is a fuzzy clique with~$|X| > 2k + c + 1$. 
 	Then, remove an arbitrary vertex~$v \in X$ from~$G$.  
 \end{rrule}
 
 \begin{figure}[t]
 	\begin{center}
 		\begin{tikzpicture}[xscale=0.5,yscale=0.5]
 			\tikzstyle{real} = [-]
 			\tikzstyle{fuzzy} = [dashed,line width=0.2mm]
 			\tikzstyle{non} = [dotted]
 			\node[](v1) at (0, -3) [shape = circle, draw, fill=black, scale=0.07ex]{};
			\node[](v2) at (2, -3) [shape = circle, draw, fill=black, scale=0.07ex]{};
			\node[](v3) at (1, -1) [shape = circle, draw, fill=black, scale=0.07ex]{};
 			
 			\node[](x1) at (-2, 1) [shape = circle, draw, fill=black, scale=0.07ex]{};
 			\node[](x2) at (1, 1) [shape = circle, scale=0.07ex]{};
 			\node[](x3) at (4, 1) [shape = cross out, draw, color=red, scale=0.25ex,label=left:$v$]{};
 			\node[](xx) at (1.1, 1) [shape = circle, scale=0.4ex]{$\dots$};
 			
 			\node[](y11) at (6, 4) [shape = circle, draw, fill=black, scale=0.07ex]{};
 			\node[](y12) at (8, 4) [shape = circle, scale=0.3ex]{$\dots$};
 			\node[](y13) at (10, 4) [shape = circle, draw, fill=black, scale=0.07ex]{};
 			\node[](y21) at (6, -2) [shape = circle, draw, fill=black, scale=0.07ex]{};
 			\node[](y22) at (8, -2) [shape = circle, scale=0.3ex]{$\dots$};
 			\node[](y23) at (10, -2) [shape = circle, draw, fill=black, scale=0.07ex]{};

 			\path[] (v1) edge (v2);
 			\path[] (v1) edge (v3);
 			\path[] (v2) edge (v3);
 			
 			\path[] (v1) edge (x1);
 			\path[] (v1) edge (x2);
 			\path[] (v1) edge (x3);
 			\path[] (v2) edge (x1);
 			\path[] (v2) edge (x2);
 			\path[] (v2) edge (x3);
 			\path[] (v3) edge (x1);
 			\path[] (v3) edge (x2);
 			\path[] (v3) edge (x3); 			
 			
 			\path[fuzzy] (y11) edge (x1);
 			\path[fuzzy] (y11) edge (x2);
 			\path[fuzzy] (y11) edge (x3);
 			\path[fuzzy] (y12) edge (x1);
 			\path[fuzzy] (y12) edge (x2);
 			\path[fuzzy] (y12) edge (x3);
 			\path[fuzzy] (y13) edge (x1);
 			\path[fuzzy] (y13) edge (x2);
 			\path[fuzzy] (y13) edge (x3);
 			
 			\path[non] (y21) edge (x1);
 			\path[non] (y21) edge (x2);
 			\path[non] (y21) edge (x3);
 			\path[non] (y22) edge (x1);
 			\path[non] (y22) edge (x2);
 			\path[non] (y22) edge (x3);
 			\path[non] (y23) edge (x1);
 			\path[non] (y23) edge (x2);
 			\path[non] (y23) edge (x3);

 			\node[draw,shape=ellipse,
			minimum height=20,
			minimum width=20,
			line width=0.8,
			color=blue,
			fit= (x1) (x2) (x3) 
			](e) at (1,1){};

			\node[draw,shape=ellipse,
			minimum height=20,
			minimum width=20,
			line width=0.8,
			color=blue,
			fit= (v1) (v2) (v3)  
			](e) at (1,-2){};

			\node[draw,shape=ellipse,
			minimum height=20,
			minimum width=80,
			line width=0.8,
			color=blue,
			fit= (y21) (y22) (y23)
			](e) at (8,-2){};
			
			\node[draw,shape=ellipse,
			minimum height=20,
			minimum width=80,
			line width=0.8,
			color=blue,
			fit= (y11) (y12) (y13)
			](e) at (8,4){};

 			\node[](K) at (-2,-2) {$K$};
 			\node[](F) at (-4.5,1) {$X$};
 			
 			\node[](V1) at (12,4) {$V_1$};
 			\node[](V2) at (12,-2) {$V_2$};
 			
 		\end{tikzpicture}
 	\end{center}
 	\caption{An illustration of \Cref{rrule:bound fuzzy clique in pure clusters}.
 		Fuzzy edges are dashed, nonedges are dotted. Only relevant edges are drawn. Vertices in~$V_1$ are fuzzy neighbors of all vertices in~$X$, vertices in~$V_2$ are not part of the same cluster as~$K \cup X$.
 		}
 	\label{fig:bounded fuzzy clique}
 \end{figure}
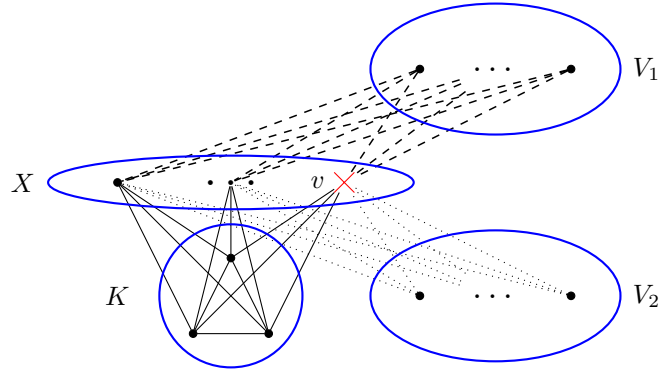
 
 \begin{lemma}
 	\Cref{rrule:bound fuzzy clique in pure clusters} is correct and can be exhaustively applied in~$\Oh(n^4) $ time.
 \end{lemma}

\begin{proof}
	Let~$\mc$ be a solution, so~$\cost(\mc) \leq k$.
	Note that since~$|X| > k$ each pair of vertices in~$K$ has at least~$k+1$ common neighbors in~$X$, so all vertices of~$K$ are part of the same cluster~$C_K$ in~$\mc$ according to~\Cref{obs:many common neighbors}.
	Moreover, since $|X| > 2k+c $, there is a subset~$X_K$ of size at least~$k+c$ that also belongs to the same cluster~$C_K$.
	
	We first show that we can assume~$X_K = X$:
	Suppose that there is some vertex~$u\in X \setminus X_K$ that is not in the cluster~$C_K$, but in some other cluster~$C_u$.
	Consider the clustering~$\mc' \coloneqq (\mc \setminus \{C_K, C_u\}) \cup \{C_K\cup \{u\}, C_u \setminus \{u\} \} $ where~$u$ is assigned to~$C_K$ instead.
	Since~$u$ is adjacent to each vertex in~$K$, but has no real neighbors outside of~$K$, and~$X$ is a fuzzy clique we have~
	\begin{align*}
			\cost(\mc) - \cost(\mc') & =  \omega(E^+(u,C_K))+\omega(E^-(u,C_u\setminus \{u\})) - \omega(E^+(u,C_u\setminus \{u\})) -\omega(E^-(u,C_K) \\
			& \geq |K| + \omega(E^-(u,C_u\setminus \{u\})) - \omega(E^-(u,C_K\setminus \{K\cup X\})) \\
			& \geq |K| - \omega(E^-(u,C_K\setminus \{K\cup X\}))
	\end{align*}
	We show that~$E^-(u,C_K\setminus \{K\cup X\}) = \emptyset$, which implies~$\cost(\mc) - \cost(\mc') > 0$.
	Consider some vertex~$w\in V \setminus (K \cup X)$.
	Per definition of~$X$, there is no real edge between~$w$ and vertices in~$X$.
	Thus, we can distinguish between two cases:
	\begin{enumerate}
		\item If~$|N^0(w)\cap X| >  c$, then for each~$x\in X$ we have that~$w$ and~$x$ have more than~$c$ common fuzzy neighbors in~$X$, which implies $N^0(w)\cap X = X$.
		\item Otherwise, we have~$|N^-(w)\cap X| >  2k$ and thus~$|N^-(w)\cap X_K| > k$.
		This implies that~$w$ is not in the cluster~$C_K$.
	\end{enumerate}
	
	Let~$V_1$ be the set of vertices~$w \in V \setminus (K \cup X)$ with~$|N^0(w)\cap X| >  c$ and let~$V_2$ be the set of vertices~$w\in V \setminus (K \cup X)$ with~$|N^-(w)\cap X_K| >  k$.
	Since each vertex in~$V_1$ is a fuzzy neighbor of~$u$ and each vertex in~$V_2$ is not part of~$C_K$, we have~$E^-(u,C_K\setminus \{K\cup X\}) = \emptyset$ and thus~$\cost(\mc) - \cost(\mc') > 0$. Hence,~$\mc'$ is also a solution with~$\cost(\mc')<k$ and we can assume~$X_K = X$.
	
	Now, let~$G'$ be the fuzzy graph obtained from~$G$ after removing some arbitrary~$v\in X$ from~$G$ and let~$X' = X \setminus \{v\}$.
	We show that~$G$ has a clustering~$\mc$ with~$\cost(\mc) \leq k$ if and only if~$G'$ has a clustering~$\mc'$ with~$\cost(\mc') \leq k$.
	
	First, suppose $\mc$ is a clustering for~$G$ with~$\cost(\mc) \leq k$.
	Since~$G'$ is an induced subgraph of~$G$, we have~$\cost(\mc') \leq \cost(\mc) \leq k$.
	
	Second, suppose that $\mc'$ is a clustering for~$G'$ with~$\cost(\mc') \leq k$.
	With the same arguments as above, since $|X'|=|X|-1$, we can see that all the vertices~$K \cup X'$ are part of the same cluster, which we denote by~$C'_K$.
	Consider the clustering~$\mc$ for~$G$ that is identical to~$\mc'$ except that the cluster~$C'_K$ is replaced by~$C_K = C'_K \cup \{v\}$.
	We have~\begin{align*}
		\cost(\mc) - \cost(\mc') & =  \omega(E^-(v,C_K')) + \omega(E^+(v,V \setminus C_K')).
	\end{align*}
	Note that~$v$ has no real edges to vertices outside of~$K \cup X'$ per definition of~$X$, so~$E^+(v,V \setminus C_K') = \emptyset$.
	Moreover,~$v$ has no nonedges to vertices of~$K \cup X' \subseteq C'_K$. 
	For each vertex~$u \in V_1$ we have~$\{u,v\} \in E^0$.
	Furthermore, each vertex~$u \in V_2$ is not part of the cluster~$C'_K$ due to the same arguments as above.
	Therefore, we also have~$E^-(v,C_K')=\emptyset$ and thus~$\cost(\mc)=\cost(\mc') \leq k$.
	
	Concerning the running time, we consider each vertex~$v \in V$, determine~$K =N^+(v)$ and~$X = N_{G^{(+)}}^{\cap}(K)$ and then check if~$K$ and~$X$ satisfy the requirements of the rule in~$\Oh(n^3)$ time.
	Since each rule application removes a vertex, we can exhaustively apply the rule in~$\Oh(n^4)$ time.
\end{proof}

 We can now bound the size of a cluster in a solution. 

  \begin{lemma}
  	\label{lem:size of a cluster bound}
	Let~$(G,k)$ be a yes-instance that is reduced with respect to \Cref{rrule:remove isolated cliques} and Rules \ref{rrule:many common neighbors closure} to~\ref{rrule:bound fuzzy clique in pure clusters}.
	Let~$\mc$ be a solution for~$(G,k)$ and let~$C\in \mc$ be a cluster.
	Then,~$|C| \leq 3kc^2 + 3kc + 3k + 3c + 3$.
\end{lemma}

\begin{proof}
	Since~$G$ is reduced with respect to~\Cref{rrule:not many real and fuzzy edges} we can partition~$C$ into~$C_1 \coloneqq C\cap V_0$ and~$C_2 \coloneqq C\cap V_+$.
	Assume towards a contradiction that~$|C| > 3kc^2 + 3kc + 3k + 3c + 3$.
	Then $C_1$ is a fuzzy clique and~$C_2$ is a real clique according to~\Cref{lem:clique and fuzzy clique in pure cluster}.
	In particular,~$C$ does not contain a nonedge.
	Since~$G$ is reduced with respect to~\Cref{rrule:bounded cliques c-closure} we can assume~$|C_2| < 3kc+2$.
	Moreover, we can bound the number of vertices in~$C_1$ that have a fuzzy neighbor in~$C_2$ by~$|C_2| \cdot c < 3kc^2+2c$.
	Note that we can assume that~$C_2$ is not empty according to~\Cref{obs:cluster connected by real edges}.
	Now, let~$C'_1 \subseteq C_1$ be the vertices in~$C_1$ that have only real neighbors in~$C_2$. 
	Since~$\mc$ is a solution, at most~$k$ vertices of~$C'_1$ have real neighbors outside of~$C$.
	Finally, let~$X\subseteq C'_1$ be the vertices in~$C'_1$ that have only real neighbors in~$C_2$ and no real neighbors outside of~$C$.
	Note that since~$X$ is a fuzzy clique, if~$|X| > 2k + c +1$ we could apply \Cref{rrule:bound fuzzy clique in pure clusters} to~$X$.
	Since~$G$ is reduced with respect to~\Cref{rrule:bound fuzzy clique in pure clusters}, we thus have~$|X| \leq 2k+c +1$.
	Overall, we get
	\begin{align*}
		|C| & = |C_1| + |C_2| \\
		& \leq  (3kc^2 +2c) + k + (2k + c +1) + (3kc + 2) \\
		& = 3kc^2 + 3kc + 3k + 3c + 3.
	\end{align*}
\end{proof}
With the previous lemma we now can bound the size of a yes-instance.

 \begin{lemma}
	Let~$I \coloneqq (G,k)$ be an instance of \CC that is reduced with respect to \Cref{rrule:remove isolated cliques} and Rules \ref{rrule:many common neighbors closure} to~\ref{rrule:bound fuzzy clique in pure clusters}.
	If~$|V(G)| > 30k^2c^2$, then~$I$ is a no-instance.
\end{lemma}

\begin{proof}
	Let~$I$ be a yes-instance and let~$\mc$ be a solution. 
	Since~$I$ is reduced with respect to \Cref{rrule:remove isolated cliques}, each cluster in~$\mc$ contributes to~$\cost(\mc)$.
	Since each deletion of a real edge can create at most~$2$ clusters,~$\mc$ contains at most~$2k$ clusters.
	According to \Cref{lem:size of a cluster bound}, if~$I$ is a yes-instance, then each cluster has a size of at most $3kc^2 + 3kc + 3k + 3c + 3$.
	Hence, either~$I$ is a no-instance or~$G$ contains at most~$6k^2c^2 + 6k^2c + 6k^2 + 6kc + 6k \leq 30k^2c^2$ vertices.
\end{proof}

The previous lemma immediately gives us \Cref{thm:closure kernel}.

\closurekernel*

\section{Hardness for Restricted Real Edges and Nonedges}
In this section we assume that all real and nonedges have unit weight.

\subsection{Kernelization Hardness for Bounded Degeneracy and Closure}

We first adapt a reduction of Demaine et al.~\cite{DemaineEFI06} from \MC to \CC to show hardness for the case where the nonedges form a matching.
\prob{\MC}{A graph~$G=(V,E)$, a set of terminal pairs~$T=\{(s_1,t_1),\dots,(s_r,t_r)\}\subseteq {V \choose 2}$ and a nonnegative integer $k$.}{Is there a set of edges~$S\subseteq E$ such that~$|S| \leq k$ and no terminal pair~$(s_i,t_i)$ is connected in~$G - S$?}

\subparagraph{Construction.}
Let~$(G=(V,E),T,k)$ with~$T=\{(s_1,t_1),\dots,(s_r,t_r)\}$ be an instance of \MC with maximum degree~$\Delta$.
We construct an instance~$(G'=(V',E^+,E^0),k)$ of \CC as follows. 
We initialize~$G'$ with~$V' = V$ and~$E^+ = E, E^0 = {V \choose 2} \setminus E$.
For each terminal vertex~$v\in V(T)$ let~$T_v$ denote the set of terminals~$u\in V(T)$ that are in a terminal pair together with~$v$ and let~$t^1_v,\dots,t^{|T_v|}_v$ be an arbitrary fixed ordering of~$T_v$.
For each terminal vertex~$v\in V(T)$ we add the following gadget~$Q_v$ consisting of a chain of~$|T_v|$ cliques of size~$\Delta+1$:
For each~$i \in \{1,\dots,|T_v|\}$ we have a clique~$Q^i_v$ of size~$\Delta+1$ associated with the vertex~$t^i_v$, where each vertex of~$Q^1_v$ is a real neighbor of~$v$ and for each~$i \in \{1,\dots,|T_v|-1\}$ every vertex in~$Q^i_v$ is a real neighbor of every vertex in~$Q^{i+1}_v$.
Note that~$Q_v\cup \{v\}$ is~$\Delta+1$-edge-connected, that is, after deleting up to~$\Delta$ edges, every vertex in~$Q_v\cup \{v\}$ is still connected to every other vertex in~$Q_v\cup \{v\}$.
For each terminal pair~$(u,v)\in T$ each vertex in~$Q^i_v$ is a nonneighbor of exactly one vertex in~$Q^j_u$ and vice versa, where~$u = t^i_v$  and~$v= t^j_u$.
In other words,~$E^-(Q^i_v, Q^j_u)$ is a matching of size~$\Delta+1$ in~$G'^{(-)}$.
All other edges adjacent to some vertex in a clique~$Q^i_v$ are fuzzy edges.

\begin{lemma}
	\label{lem:reduction for bounded real degeneracy and closure }
	$(G,k)$ is a yes-instance of \MC if and only if $(G',k)$ is a yes-instance of \CC.
\end{lemma}

\begin{proof}
	First, let $(G,k)$ be a yes-instance with solution set~$S$ and let~$C_1,\dots,C_{\ell}$ be the connected components of~$G-S$.
	Consider the clustering~$\mc = \{C'_1,\dots,C'_{\ell}\}$, where each~$C'_i$ consists of the vertices in~$C_i$ together with all vertices of~$Q_v$ for each terminal vertex~$v\in V(T)$ in~$C_i$.
	Note that the only real edges between clusters in~$\mc$ are the edges in~$S$. 
	Moreover, since~$S$ is a solution for \MC, no connected component in~$G-S$ contains both vertices of a terminal pair, so no cluster in~$\mc$ contains a nonedge.
	Thus,~$\cost(\mc) = |S| \leq k$ and~$(G',k)$ is a yes-instance.
	
	Now, let $(G',k)$ be a yes-instance with an optimal clustering~$\mc$ with~$\cost(\mc)\leq k$.
	We show the following claim:
	\begin{claim}
		For each~$v\in V(T)$ all vertices of~$Q_v \cup \{v\}$ are in the same cluster in~$\mc$ and no cluster in~$\mc$ contains vertices from gadgets~$Q_v \cup \{v\}$ and~$Q_u \cup \{u\}$ of a terminal pair~$(u,v) \in T$.
	\end{claim}
	
	\begin{claimproof}
		Assume towards a contradiction that for some~$v\in V(T)$ there are clusters~$C_1,\dots,C_\ell$ that each contain vertices of~$Q_v\cup \{v\}$.
		Consider the new clustering~$\mc' = (\mc \setminus \{C_1,\dots,C_\ell\}) \cup \{Q_v \cup \{v\}, C_1\setminus (Q_v \cup \{v\}), \dots, C_\ell \setminus (Q_v \cup \{v\})\}$, where~$Q_v \cup \{v\}$ becomes its own cluster.
		Since~$Q_v \cup \{v\}$ is~$\Delta+1$-edge-connected, at least~$\Delta+1$ edges in~$E^+(Q_v \cup \{v\})$ have to be deleted to obtain~$\mc$, whereas disconnecting~$Q_v\cup \{v\}$ from the other clusters in~$\mc'$ needs~$\deg_G(v) \leq \Delta$ edge deletions.
		Thus,~$\cost(\mc') < \cost(\mc)$, a contradiction to the optimality of~$\mc$.
		
		Now, assume towards a contradiction that some cluster~$C\in \mc$ contains the gadgets~$Q_v \cup \{v\}$ and~$Q_u \cup \{u\}$ for a terminal pair~$(u,v) \in T$.
		Consider the new clustering~$\mc' = (\mc \setminus \{C\}) \cup  \{Q_v\cup \{v\}, C\setminus  (Q_v\cup \{v\} )\}$, where~$Q_v \cup \{v\}$ becomes its own cluster.
		We have
		\begin{align*}
			\cost(\mc) - \cost(\mc')  & =  E^-(Q_v\cup \{v\}, C \setminus (Q_v\cup \{v\})) - E^+(Q_v\cup \{v\}, C \setminus (Q_v\cup \{v\})) \\
			& \geq  E^-(Q_v\cup \{v\}, Q_u\cup \{u\}) - \deg_G(v)\\
			& \geq  (\Delta+1) -\Delta > 0
		\end{align*}
		and thus a contradiction to the optimality of~$\mc$. 
	\end{claimproof}
	By the above claim, only real edges  between clusters in~$\mc$ contribute to~$\cost(\mc)$, since per construction nonedges are only present between vertices of gadgets~$Q_u$ and $Q_v$ for some terminal pair $(u,v)$.
	Thus, if we consider the set~$S = \{e \in E^+(C_1,C_2) \mid C_1,C_2 \in \mc,C_1 \neq C_2\}$ of real edges between clusters in~$\mc$, we have~$|S| \leq k$.
	Since for each terminal pair~$(u,v)$ the gadgets~$Q_u \cup \{u\}$ and~$Q_v \cup \{v\}$ are in different clusters of~$\mc$, after deleting the edges of~$S$ in~$G$, no connected component in~$G-S$ contains both vertices of a terminal pair and thus $S$ is a solution for $(G,k)$.
\end{proof}

We now describe a reduction from~\CC where the nonedges induce a matching to~\CC with the additional constraint that the closure and the degeneracy of~$G^{(+)}$ are at most 2.

\subparagraph{Construction.}
Let $(G=(V,E^+,E^0),k)$ be an instance of \CC where the nonedges induce a matching.
We construct a new instance~$(G'=(V',\widetilde{E}^+,\widetilde{E}^0),k)$ of \CC as follows.
We subdivide each real edge~$\{u,v\}$ by adding a new vertex~$x_{uv}$ and making~$\{u,x_{uv}\}$ as well as~$\{v,x_{uv}\}$ real edges. Moreover, all other edges incident with~$x_{uv}$ as well as the edge $\{u,v\}$ become fuzzy edges.
Note that the set of nonedges did not change and thus the nonedges in~$G'$ also induce a matching.
Furthermore, since no vertex pair in~$G'$ has more than one common real neighbor and after deleting each new vertex~$x_{uv}$ with real degree 2 all other vertices have no real neighbors anymore,~$G'^{(+)}$ has closure 2 and degeneracy 2.

\begin{lemma}
	\label{lem:reduction subdivide real edges}
	$(G,k)$ is a yes-instance of \CC if and only if $(G',k)$ is a yes-instance of \CC.
\end{lemma}

\begin{proof}
	Let $(G,k)$ be a yes-instance with a clustering~$\mc$ of cost at most~$k$.
	Consider the clustering~$\mc'$ for~$G'$ where for each cluster~$C \in \mc$ we have a cluster~$C'\in \mc'$ such that~$C\subseteq C'$ and for each new vertex~$x_{uv}$ we have~$x_{uv} \in \mc'(u)$ or~$x_{uv} \in \mc'(v)$ arbitrarily.
	Note that the contribution of nonedges to~$\cost(\mc')$ is the same as the contribution of nonedges to~$\cost(\mc)$.
	Now, if some real edge~$\{u,x_{uv}\}$ contributes to~$\cost(\mc')$, then this implies that~$x_{uv} \in \mc'(v)$ and~$\mc'(u) \neq \mc'(v)$, so in particular~$\mc(u) \neq \mc(v)$ and~$\{u,v\}$ contributes to~$\cost(\mc)$.
	Thus,~$\cost(\mc') \leq \cost(\mc) = k $.
	
	Conversely, let $(G',k)$ be a yes-instance with a clustering~$\mc'$ of cost at most~$k$.
	Note that we can assume that each vertex~$x_{uv}$ is in the cluster~$\mc'(u)$ or~$\mc'(v)$.
	Otherwise, since~$x_{uv}$ has only fuzzy neighbors except for~$u$ and~$v$, we could put~$x_{uv}$ in the cluster~$\mc'(u)$ instead and thus reduce the cost of the clustering by not deleting~$\{u,x_{uv}\}$.
	Consider the clustering $\mc$ for~$G$ where for each cluster~$C'\in \mc'$ we have a cluster~$C \in \mc$ with~$C = C' \cap V$.
	Again, the contribution of nonedges to~$\cost(\mc)$ is the same as the contribution of nonedges to~$\cost(\mc')$.
	Moreover, if some real edge~$\{u,v\}$ contributes to~$\cost(\mc)$, then this implies~$\mc(u) \neq \mc(v)$ and thus~$\mc'(u) \neq \mc'(v)$. 
	In particular, we have either~$x_{uv} \notin \mc'(u)$ or~$x_{uv} \notin \mc'(v)$ and thus either~$\{u,x_{uv}\}$ or~$\{v,x_{uv}\}$ contributes to~$\cost(\mc')$.
	Hence,~$\cost(\mc) \leq \cost(\mc') = k $.
\end{proof}

Since \MC is \NP-hard and does not admit a polynomial kernel for the parameter~$k$, the reductions above which preserve the parameter value imply the following. 

\kernelhardness*

\subsection{Hardness on Binary Trees and Stars}
\MC has been shown to be NP-hard on binary trees by reduction from \TSAT~\cite{CalinescuFR03}. We slightly adapt this reduction and combine it with the reduction presented of Demaine et al. \cite{DemaineEFI06}  from \MC to \CC to show that \CC is NP-hard even when~$G^{(+)}$ is a binary tree.

First we restate the reduction from \TSAT to \MC on binary trees together with a slight modification to ensure that each terminal vertex is a leaf.

\prob{\TSAT}{A CNF-formula $\Phi$ where each clause contains exactly three literals.}{Is there a satisfying assignment for $\Phi$?}

\subparagraph{Construction.} 
Let~$\Phi$ be an instance of \TSAT with variables~$x_1,\dots,x_n$ and clauses~$C_1,\dots C_m$.
We construct an instance~$(G,T,k)$ of \MC as follows:

For each variable~$x_i$ we create a variable gadget~$H_i$ that is a binary tree consisting of two leaf vertices labeled with~$x_i$ and~$\bar{x}_i$, respectively, and a common parent vertex.
For each clause~$C_j=(x \vee y \vee z)$ we create a clause gadget~$F_j$ that is a binary tree consisting of four leaf vertices, three of which are labeled with~$x$,~$y$ and~$z$, and three internal vertices, as seen in figure \Cref{fig:construction 1}.
All the vertex and clause gadgets are then arbitrarily combined into a binary tree. 
The set of terminal pairs~$T$ contains for each variable~$x_i$ the pair of vertices labeled with~$x_i$ and $\bar{x}_i)$ in the variable gadget of~$x_i$.
Moreover, for each clause~$C_j = (x \vee y \vee z)$ we add a terminal pair~$(x,y)$ and a terminal pair~$(z,w)$, where without loss of generality~$x$ and~$y$ are the leaves that share the same parent vertex and~$w$ is the fourth leaf that does not correspond to a variable.
Finally, each vertex labeled with the literal~$\tilde{x}_i$ in a clause gadget is in a terminal pair with the vertex labeled~$\tilde{x}_i$ in the vertex gadget for~$x_i$, where~$\tilde{x}_i \in \{x_i,\bar{x}_i\}$.
We finish the construction by setting~$k \coloneqq n + 2m$.

\begin{lemma}
	\label{lem:reduction multicut on binary trees}
	$\Phi$ is satisfiable if and only if~$(G,T,k)$ is a yes-instance.
\end{lemma}

\begin{proof}
	Let~$\alpha$ be a satisfying assignment for~$\Phi$.
	We define a solution set~$S$ for~$(G,T,k)$ as follows.
	For each variable~$x_i$ the set~$S$ contains one edge of~$H_i$: the edge incident to the vertex labeled~$x_i$, if~$\alpha(x_i) = \true$, or the edge incident to the vertex labeled~$\bar{x}_i$, if~$\alpha(x_i)=\false$.
	For a clause~$C_j=(x \vee y \vee z)$ we put two edges of~$F_j$ into~$S$, such that both terminal pairs are disconnected and for the one terminal vertex still connected to the root of~$F_j$, say~$x$, we have that~$x$ is true under~$\alpha$.
	Clearly, we have~$|S| = n+2m$ and each terminal pair inside a variable or vertex gadget is disconnected by~$S$.
	Moreover, since for each clause the only terminal vertex still connected to the root of the gadget is one, for which the literal evaluates to~$\true$ under~$\alpha$ and the edge incident to that literal in the respective variable gadget is in~$S$, all terminals in~$T$ are disconnected in~$G-S$ and thus~$(G,T,k)$ is a yes-instance.
	
	Now, let~$(G,T,k)$ be a yes-instance with solution set~$S$.
	Note that for each vertex gadget~$H_i$ one of the two edges must be in~$S$ and for each clause gadget~$F_j$ at least two of the edges must be in~$S$, otherwise there is some terminal pair that is still connected in~$G-S$.
	Since~$|S| \leq k = n+2m$ in fact $S$ contains exactly two edges for each clause gadget, which implies that in each clause gadget~$F_j$ there is one vertex that is still connected to the root of~$F_j$, which we denote with~$z_j$.
	Let~$F_j$ be an arbitrary clause gadget and let~$\tilde{x_i}$ for some~$\tilde{x}_i \in \{x_i,\bar{x}_i\}$ be the label of~$z_j$. 
	Since~$z_j$ is in a terminal pair with the vertex of~$H_i$ with label~$\tilde{x_i}$, the set~$S$ must contain the edge incident to the vertex in~$H_i$ with label~$\tilde{x_i}$.
	This implies that the assignment~$\alpha$, where for each variable~$x_i$ we have~$\alpha(x_i)=\true$ if $S$ contains the edge incident to the vertex in~$H_i$ with label~$x_i$ and~$\alpha(x_i)=\false$ otherwise, is a satisfying assignment for~$\Phi$.
\end{proof}

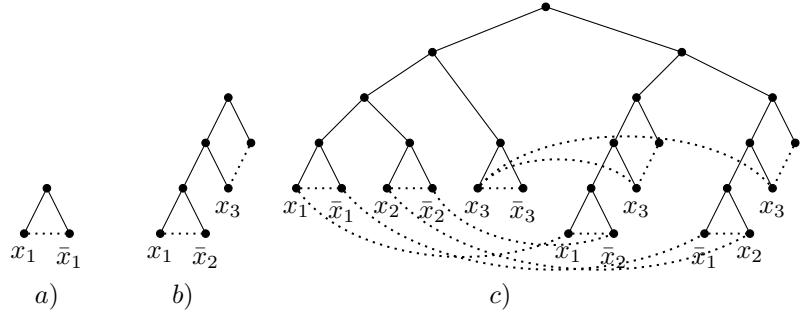
\begin{figure}[t]
	\begin{center}
		\begin{tikzpicture}[xscale=0.6,yscale=0.4]
			\tikzstyle{real} = [-]
			\tikzstyle{non} = [dotted,line width=0.75pt]
			\node[label=below:$x_1$](x1) at (0, 0) [shape = circle, draw, fill=black, scale=0.07ex]{};
			\node[label=below:$\bar{x}_1$](x2) at (1, 0) [shape = circle, draw, fill=black, scale=0.07ex]{};
			\node[](u1) at (0.5, 1.5) [shape = circle, draw, fill=black, scale=0.07ex]{};
			\node[label=below:$a)$](a) at (0.5, -1) []{};
			
			\path [real] (x1) edge (u1);
			\path [real] (x2) edge (u1);		
			
			\node[label=below:$x_1$](c1) at (3, 0) [shape = circle, draw, fill=black, scale=0.07ex]{};
			\node[label=below:$\bar{x}_2$](c2) at (4, 0) [shape = circle, draw, fill=black, scale=0.07ex]{};
			\node[](w1) at (3.5, 1.5) [shape = circle, draw, fill=black, scale=0.07ex]{};
			\node[](w2) at (4, 3) [shape = circle, draw, fill=black, scale=0.07ex]{};
			\node[](w3) at (5, 3) [shape = circle, draw, fill=black, scale=0.07ex]{};
			\node[](w4) at (4.5, 4.5) [shape = circle, draw, fill=black, scale=0.07ex]{};
			\node[label=below:$x_3$](c3) at (4.5, 1.5) [shape = circle, draw, fill=black, scale=0.07ex]{};
			\node[label=below:$b)$](b) at (3.5, -1) []{};			
			
			\path [real] (c1) edge (w1);
			\path [real] (c2) edge (w1);
			\path [real] (w1) edge (w2);
			\path [real] (w2) edge (c3);
			\path [real] (w2) edge (w4);
			\path [real] (w4) edge (w3);
			
			\path [non] (x1) edge (x2);
			\path [non] (c1) edge (c2);
			\path [non] (w3) edge (c3);
			
			\node[label=below:$x_1$](x11) at (6, 1.5) [shape = circle, draw, fill=black, scale=0.07ex]{};
			\node[label=below:$\bar{x}_1$](x12) at (7, 1.5) [shape = circle, draw, fill=black, scale=0.07ex]{};
			\node[](u11) at (6.5, 3) [shape = circle, draw, fill=black, scale=0.07ex]{};

			\node[label=below:$x_2$](x21) at (8, 1.5) [shape = circle, draw, fill=black, scale=0.07ex]{};
			\node[label=below:$\bar{x}_2$](x22) at (9, 1.5) [shape = circle, draw, fill=black, scale=0.07ex]{};
			\node[](u21) at (8.5, 3) [shape = circle, draw, fill=black, scale=0.07ex]{};

			\node[label=below:$x_3$](x31) at (10, 1.5) [shape = circle, draw, fill=black, scale=0.07ex]{};
			\node[label=below:$\bar{x}_3$](x32) at (11, 1.5) [shape = circle, draw, fill=black, scale=0.07ex]{};
			\node[](u31) at (10.5, 3) [shape = circle, draw, fill=black, scale=0.07ex]{};

			\node[label=below:$x_1$](c11) at (12, 0) [shape = circle, draw, fill=black, scale=0.07ex]{};
			\node[label=below:$\bar{x}_2$](c12) at (13, 0) [shape = circle, draw, fill=black, scale=0.07ex]{};
			\node[](w11) at (12.5, 1.5) [shape = circle, draw, fill=black, scale=0.07ex]{};
			\node[](w12) at (13, 3) [shape = circle, draw, fill=black, scale=0.07ex]{};
			\node[](w13) at (14, 3) [shape = circle, draw, fill=black, scale=0.07ex]{};
			\node[](w14) at (13.5, 4.5) [shape = circle, draw, fill=black, scale=0.07ex]{};
			\node[label=below:$x_3$](c13) at (13.5, 1.5) [shape = circle, draw, fill=black, scale=0.07ex]{};
			\node[label=below:$c)$](c) at (10.5, -1) []{};
			
			\node[label=below:$\bar{x}_1$](c21) at (15, 0) [shape = circle, draw, fill=black, scale=0.07ex]{};
			\node[label=below:$x_2$](c22) at (16, 0) [shape = circle, draw, fill=black, scale=0.07ex]{};
			\node[](w21) at (15.5, 1.5) [shape = circle, draw, fill=black, scale=0.07ex]{};
			\node[](w22) at (16, 3) [shape = circle, draw, fill=black, scale=0.07ex]{};
			\node[](w23) at (17, 3) [shape = circle, draw, fill=black, scale=0.07ex]{};
			\node[](w24) at (16.5, 4.5) [shape = circle, draw, fill=black, scale=0.07ex]{};
			\node[label=below:$x_3$](c23) at (16.5, 1.5) [shape = circle, draw, fill=black, scale=0.07ex]{};
			
			\node[](u12) at (7.5, 4.5) [shape = circle, draw, fill=black, scale=0.07ex]{};
			\node[](u123) at (9, 6) [shape = circle, draw, fill=black, scale=0.07ex]{};
			\node[](uc1c2) at (14.5, 6) [shape = circle, draw, fill=black, scale=0.07ex]{};
			\node[](r) at (11.5, 7.5) [shape = circle, draw, fill=black, scale=0.07ex]{};
			
			\path [real] (x11) edge (u11);
			\path [real] (x12) edge (u11);
			\path [real] (x21) edge (u21);
			\path [real] (x22) edge (u21);	
			\path [real] (x31) edge (u31);
			\path [real] (x32) edge (u31);
			\path [real] (u11) edge (u12);
			\path [real] (u21) edge (u12);		
			\path [real] (u12) edge (u123);
			\path [real] (u31) edge (u123);
			\path [real] (w14) edge (uc1c2);
			\path [real] (w24) edge (uc1c2);
			\path [real] (u123) edge (r);
			\path [real] (uc1c2) edge (r);
			
			\path [real] (c11) edge (w11);
			\path [real] (c12) edge (w11);
			\path [real] (w11) edge (w12);
			\path [real] (w12) edge (c13);
			\path [real] (w12) edge (w14);
			\path [real] (w14) edge (w13);
			
			\path [real] (c21) edge (w21);
			\path [real] (c22) edge (w21);
			\path [real] (w21) edge (w22);
			\path [real] (w22) edge (c23);
			\path [real] (w22) edge (w24);
			\path [real] (w24) edge (w23);
			
			\path [non] (x11) edge (x12);
			\path [non] (x21) edge (x22);
			\path [non] (x31) edge (x32);
			\path [non] (c11) edge (c12);
			\path [non] (w13) edge (c13);
			\path [non] (c21) edge (c22);
			\path [non] (w23) edge (c23);			
			
			\path [bend right=50,real,line width=0.3mm, non] (x11) edge (c11);
			\path [bend right=50,real,line width=0.3mm, non] (x12) edge (c21);
			\path [bend right=50,real,line width=0.3mm, non] (x21) edge (c22);
			\path [bend right=50,real,line width=0.3mm, non] (x22) edge (c12);
			\path [bend left=60,real,line width=0.3mm, non] (x31) edge (c13);
			\path [bend left=60,real,line width=0.3mm, non] (x31) edge (c23);
		\end{tikzpicture}
	\end{center}
	\caption{Example of a) a variable gadget~$H_1$ for a variable~$x_1$, b) a clause gadget~$F_1$ for a clause~$C_1 = (x_1 \vee \bar{x}_2 \vee x_3)$ and a graph~$G$ for the formula~$\Phi = (x_1 \vee \bar{x}_2 \vee x_3)\wedge (\bar{x}_1 \vee x_2 \vee x_3)$. Terminal pairs are marked by dotted lines.}
	\label{fig:construction 1}
\end{figure}

We now present a reduction from \MC on binary trees, where each terminal is a leaf, to \CC that is similar to the reduction by Demaine et al. \cite{DemaineEFI06}: 

\subparagraph{Construction.}
Let~$(G=(V,E),T,k)$ with~$T=\{(s_1,t_1),\dots,(s_r,t_r)\}$ be an instance of \MC, where~$G$ is a binary tree and every terminal vertex is a leaf.
Let~$V(T)$ denote the set of vertices that are in a terminal pair of~$T$.
We construct an instance~$(G'=(V',E^+,E^0),k')$ of \CC as follows. 

We initialize~$G'$ with~$V' = V$ and~$E^+ = E, E^0 = {V \choose 2} \setminus E$.
For each terminal vertex~$v\in V(T)$ we add two new vertices~$v_1,v_2$ to~$V'$ as well as edges~$\{v,v_1\}$ and~$\{v,v_2\}$ to~$E^+$.
For each terminal pair~$(u,v)\in T$ we make~$\{u_1,v_1\}$ and~$\{u_2,v_2\}$ nonedges and all other edges incident to newly added vertices are fuzzy edges.
Note that~$G'^{(+)}$ is still a binary tree.

\begin{lemma}
	\label{lem:reduction binary trees}
	$(G,k)$ is a yes-instance of \MC if and only if $(G',k)$ is a yes-instance of \CC.
\end{lemma}

\begin{proof}
	First, let $(G,k)$ be a yes-instance with solution set~$S$ and let~$C_1,\dots,C_{\ell}$ be the connected components of~$G-S$.
	Consider the clustering~$\mc = \{C'_1,\dots,C'_{\ell}\}$, where each~$C'_i$ consists of the vertices in~$C_i$ together with the vertices~$v_1,v_2$ for each terminal vertex~$v\in V(T)$ in~$C_i$.
	Note that the only real edges between clusters in~$\mc$ are the edges in~$S$. 
	Moreover, since~$S$ is a solution for \MC, no connected component in~$G-S$ contains both vertices of a terminal pair, so no cluster in~$\mc$ contains a nonedge.
	Thus,~$\cost(\mc) = |S| \leq k$ and~$(G',k)$ is a yes-instance.
	
	Now, let $(G',k)$ be a yes-instance with an optimal clustering~$\mc$ with~$\cost(\mc)\leq k$.
	Note that we can assume that for each terminal vertex~$v\in V(T)$ both new real neighbors~$v_1$ and~$v_2$ are in the same cluster~$\mc(v)$ as~$v$, since otherwise at least one of the two real edges~$\{v,v_1\}$ or~$\ \{v,v_2\}$ has to be deleted and we could delete the edge between~$v$ and its single neighbor in~$G$ instead. 
	Moreover, if a cluster~$C\in \mc$ contains a terminal pair~$(u,v)$ and thus also~$u_1,u_2,v_1$ and~$v_2$, then we could get a better clustering~$\mc' = (\mc \setminus \{C\}) \cup \{\{u,u_1,u_2\}, (C \setminus \{u,u_1,u_2\})\}$ by deleting the real edge between~$u$ and its single neighbor in~$G$ instead of adding the non edges~$\{u_1,v_1\}$ and~$\{u_2,v_2\}$. 
	
	Hence, we know that the only edges that contribute to~$\cost(\mc)$ are real edges between clusters in~$\mc$, since per construction non edges are only present between vertices~$u_1,u_2$ and~$v_1,v_2$ for some terminal pair $(u,v)$.
	Thus, if we consider the set~$S = \{e \in E^+(C_1,C_2) \mid C_1,C_2 \in \mc,C_1 \neq C_2\}$ of real edges between clusters in~$\mc$, we have~$|S| \leq k$.
	Since for each terminal pair~$(u,v)$ the vertices~$u_1,u_2$ and~$v_1,v_2$ are in different clusters of~$\mc$, after deleting the edges of~$S$ in~$G$, no connected component in~$G-S$ contains both vertices of a terminal pair and thus $S$ is a solution for $(G,k)$.
\end{proof}

\begin{theorem}
	\CC is \NP-hard even if the graph consisting of the real edges is a binary tree.
\end{theorem}

\begin{corollary}
	\CC is \NP-hard even if the graph consisting of the real edges has degeneracy~$1$.
\end{corollary}

We now present a reduction from \CL to \CC to show that \CC remains hard even if~$G^{(+)}$ is a star.

\subparagraph{Construction:}
Let~$(G=(V,E),k)$ be an instance of \CL with~$|V|=n$.
We construct from~$(G,k)$ an instance~$(G'=(V',E^+,E^0),k')$ of \CC as follows.
We introduce a new vertex~$v^*$ and set~$V' = V \cup \{v^*\}$.
For each edge~$\{u,v\}\in E$ we add the pair~$\{u,v\}$ as a fuzzy edge to~$E^0$ and we set each pair~$\{u,v\} \in {V \choose 2}\setminus E$ as a nonedge in~$G'$.
Moreover, we add an edge~$\{u,v^*\}$ for each~$u\in V$ to~$E^+$.
We finish the construction by setting~$k' \coloneqq n-k$.
Note that~$G'^{(+)}$ is a star with center vertex~$v^*$, so~$G'^{(+)}$ has degeneracy~$d=1$ and~$h$-index~$h=1$. 

\begin{lemma}
	\label{lem:reduction for degeneracy two}
	$(G,k)$ is a yes-instance of \CL if and only if~$(G',k')$ is a yes-instance of \CC.
\end{lemma}

\begin{proof}
	Let~$(G,k)$ be a yes-instance of \CL and let~$S\subseteq V$ be a clique of size~$k$ in~$G$.
	Consider the clustering~$\mc$ consisting of the cluster~$C^* \coloneqq S\cup \{v^*\}$ and clusters~$C_u \coloneqq \{u\}$ for each~$u\in V' \setminus (S\cup \{v^*\})$.
	Since~$S$ is a clique in~$G$, each edge between two vertices of~$S$ in~$G'$ is a fuzzy edge and thus no edge has to be inserted for the cluster~$C^*$. 
	For each cluster~$C_u$, the edge~$\{u,v^*\}$ has to be deleted, so~$\cost(\mc) = n-k = k'$ and thus~$(G',k')$ is a yes-instance of \CC.
	
	Conversely, let~$(G',k')$ be a yes-instance of \CC and let~$\mc$ be a clustering for~$G'$ with~$\cost(\mc)\leq k' = n-k$.
	Let~$C^*\in \mc$ be the cluster that contains the vertex~$v^*$.
	Each real edge between~$v^*$ and~$V\setminus C^*$ has to be deleted.
	Thus,~$C^*$ must contain at least~$k$ vertices of~$V$, since~$\cost(\mc) \leq n-k$.
	Now suppose that~$\{u,w\}$ is a nonedge in~$G'[C^*\setminus\{v^*\}]$. 
	Consider the clustering~$\widetilde{\mc} = \mc \setminus \{C^*\}\cup \{C^* \setminus \{w\}, \{w\}\}$.
	To obtain~$\widetilde{\mc}$ from~$C^*$, the real edge~$\{v^*,w\}$ has to be deleted, but the nonedge~$\{u,w\}$ must not be added.
	Thus, we have~$\cost(\widetilde{\mc}) \leq \cost(\mc)$, which implies that we can assume that $G'[C^*\setminus\{v^*\}]$ does not contain a nonedge.
	In other words, each edge between vertices of~$C^*\setminus \{v^*\}$ is a fuzzy edge in~$G'$ and thus per construction a (real) edge in~$G$.
	Therefore,~$C^*\setminus \{v^*\}$ is a clique of size at least~$k$ in~$G$ and~$(G,k)$ is a yes-instance of \CL.
\end{proof}

Since \CL is \NP-hard, \Cref{lem:reduction for degeneracy two} implies the following.

\begin{theorem}
	\label{thm:hardness h-index}
	\CC is \NP-hard even if the graph consisting of the real edges is a star.
\end{theorem}

\begin{corollary}
	\label{corr:hardness h-index}
	\CC is \NP-hard even if the graph consisting of the real edges has $h$-index~$1$.
\end{corollary}

\section{Conclusion}
\label{sec:conclusion}
Motivated by the fact that \CC~presumably does not admit a polynomial problem kernel for the solution cost parameter~$k$, we continued the investigation of the influence of structural parameters of the fuzzy edge graph on the problem complexity.
We showed that combining~$k$ with the degeneracy~$d$ or closure~$c$ of the fuzzy edge graph leads to polynomial kernels, while cases where the real edges and nonedges have a rather restrictive local structure remain hard.
A natural direction for further research would be to improve the sizes of the kernels and to find FPT-algorithms with good running times for the parameters~$k+d$ and~$k+c$.
While our results do not rule out polynomial kernels for~$k$ together with the maximum degree~$\Delta$, the presented reduction from \MC to \CC with bounded degeneracy and closure increases the maximum degree only by a constant factor.
Thus, showing that \MC does not admit a polynomial kernel for~$k+\Delta$ would also transfer the result to \CC. We leave the (non)existence of such a kernel as an open problem.

\newpage
\bibliographystyle{plainurl}

\end{document}